\documentclass[a4paper]{article}

\usepackage{authblk}
\usepackage[]{geometry}
\usepackage{hyperref}

\usepackage{graphicx}
\usepackage{graphicx}
\usepackage{natbib}
\usepackage{caption}

\usepackage[english]{babel} 
\usepackage{amsfonts} 
\usepackage{amsmath} 
\usepackage{amssymb} 
\usepackage{amsthm} 
\usepackage{dsfont} 
\usepackage{mathtools} 
\usepackage{tikz} 
\usetikzlibrary{automata, positioning} 
\usepackage{xfrac}    
\usepackage{xargs}
\usepackage{thm-restate}
\usepackage{multirow}
\usepackage{cleveref} 
\usepackage[linesnumbered,boxruled]{algorithm2e}
\IncMargin{1em}
\crefname{algocf}{Algorithm}{Algorithms}
\Crefname{algocf}{Algorithm}{Algorithms}

\newtheorem{theorem}{Theorem}

\newtheorem{lemma}[theorem]{Lemma}
\newtheorem{corollary}[theorem]{Corollary}
\newtheorem{example}{Example}
\newtheorem{definition}{Definition}
\newtheorem{remark}{Remark}

\DeclareMathOperator{\val}{val}
\DeclareMathOperator{\Parity}{Parity}
\DeclareMathOperator{\Reach}{Reach}
\DeclareMathOperator{\BeliefReach}{Belief-Reach}

\DeclareMathOperator{\priority}{pri}

\DeclareMathOperator{\supp}{supp}

\DeclareMathOperator*{\argmin}{arg\,min}
\newcommand{\transition}{\delta}
\newcommand{\policy}{\sigma}
\renewcommand{\next}[1]{#1'}
\newcommand{\another}[1]{\widetilde{#1}}
\newcommand{\POMDP}{P}
\newcommand{\MDP}{M}
\newcommand{\play}{\rho}
\newcommand{\Plays}{\Omega}
\newcommand{\Policies}{\Sigma}

\newcommand{\A}{\mathcal{A}}

\newcommand{\D}{\mathcal{D}}
\newcommand{\E}{\mathcal{E}}

\newcommand{\I}{\mathcal{I}}

\newcommand{\Q}{\mathcal{Q}}
\renewcommand{\L}{\mathcal{L}}
\newcommand{\R}{\mathcal{R}}
\renewcommand{\S}{\mathcal{S}}
\newcommand{\T}{\mathcal{T}}
\newcommand{\U}{\mathcal{U}}

\newcommand{\W}{\mathcal{W}}
\newcommand{\X}{\mathcal{X}}
\newcommand{\Z}{\mathcal{Z}}

\newcommand{\NN}{\mathbb{N}}
\newcommand{\PP}{\mathbb{P}}
\newcommand{\EE}{\mathbb{E}}

\newcommand{\eps}{\varepsilon}
\newcommand{\defas}{\coloneqq}
\newcommand{\1}{\mathds{1}}

\title{Revealing POMDPs: Qualitative and Quantitative Analysis for Parity Objectives}

\author[1]{Ali Asadi}
\author[1]{Krishnendu Chatterjee}
\author[2]{David Lurie}
\author[3]{Raimundo Saona}

\affil[1]{Institute of Science and Technology Austria}
\affil[2]{Paris Dauphine University, PSL Research University, Paris, France and NyxAir, Paris, France}
\affil[3]{London School of Ecnonomics and Political Science, London, United Kingdom}

\affil[ ]{\textit{\{ali.asadi, krishnendu.chatterjee\}@ista.ac.at, david.lurie@dauphine.eu, raimundo.saona@gmail.com}}

\begin{document}

\maketitle

\begin{abstract}
    Partially observable Markov decision processes (POMDPs) are a central model for uncertainty in sequential decision making. 
    The most basic objective is the reachability objective, where a target set must be eventually visited, and the more general parity objectives can model all $\omega$-regular specifications.
    For such objectives, the computational analysis problems are the following: 
    (a)~qualitative analysis that asks whether the objective can be satisfied with probability~$1$ (almost-sure winning) or probability arbitrarily close to~$1$ (limit-sure winning);
    and (b)~quantitative analysis that asks for the approximation of the optimal probability of satisfying the objective.
    For general POMDPs, almost-sure analysis for reachability objectives is \textnormal{EXPTIME}-complete, but limit-sure and quantitative analyses for reachability objectives are undecidable; almost-sure, limit-sure, and quantitative analyses for parity objectives are all undecidable.
    A special class of POMDPs, called revealing POMDPs, has been studied recently in several works, and for this subclass the almost-sure analysis for parity objectives was shown to be \textnormal{EXPTIME}-complete.
    In this work, we show that for revealing POMDPs the limit-sure analysis for parity objectives is \textnormal{EXPTIME}-complete, and even the quantitative analysis for parity objectives can be achieved in \textnormal{EXPTIME}.
\end{abstract}

\section{Introduction}

\paragraph{POMDPs}
\emph{Partially observable Markov decision processes} (POMDPs) model sequential decision-making under uncertainty \cite{bertsekas1976DynamicProgrammingStochastic, papadimitriouComplexityMarkovDecision1987,kaelbling1998planning}.
At each step, the environment is in some hidden state.
A controller interacts with it by choosing actions.
The chosen action and the current state determine a probability distribution over the subsequent state.
The controller cannot observe the state directly, but observes a signal that discloses only partial information of the state. 
POMDPs generalize classic models: \emph{Markov decision processes} (MDPs), in which the state is fully observed \cite{puterman1994},
and \emph{blind MDPs}, in which no state information is observed and are equivalent to Probabilistic Finite Automata \cite{rabin1963probabilistic,paz1971introduction}.

\paragraph{Objectives}
The controller aims to maximize an \emph{objective function}, which formally captures the desired behaviors of the model. 
Two main classes of objectives are typically considered: \emph{logical objectives}, e.g., reachability and LTL (linear-temporal logic), and \emph{quantitative objectives}, e.g., discounted sum and limit-average; see~\cite{puterman1994, baier2008principles, filarCompetitiveMarkovDecision1997} for details.
This work focuses on logical objectives. 

The most basic example of logical objectives is \emph{reachability}, i.e., given a set of target states, the objective requires that some target state is visited at least once.
A more general class of logical objectives is \emph{parity objectives}, which assign to each state a non-negative integer, called \emph{priority}.
The objective is satisfied if the smallest priority appearing infinitely often is even.
Parity objectives express all commonly used temporal objectives such as liveness, and are a canonical form to express all $\omega$-regular objectives~\cite{thomasLanguagesAutomataLogic1997} and LTL objectives~\cite{baier2008principles}. 
Hence, the study of POMDPs with parity objectives is a fundamental theoretical problem.

\paragraph{Applications}
POMDPs have been applied across diverse fields, including computational biology \cite{durbin1998biological} and reinforcement learning \cite{kaelbling1996reinforcement}.
In particular, POMDPs with logical objectives have proven useful in many application areas such as probabilistic planning~\cite{bonet2009solving}; randomized embedded scheduler~\cite{de2005code}; randomized distributed algorithms~\cite{pogosyants2000verification}; probabilistic specification languages~\cite{baier2012probabilistic}; and robot planning \cite{kress2009temporal,chatterjee2015qualitative}.

\paragraph{Computational Analysis Questions}
A policy determines the choice of actions by the controller. 
The value corresponds to the maximum probability that the controller can guarantee to satisfy an objective.
The main computational analysis of POMDPs with reachability and parity objectives considers the following two problems.
\begin{itemize}
    \item[(a)] 
        \emph{Qualitative analysis} has two variants: 
        the \emph{almost-sure winning} asks whether there exists a policy that satisfies the objective with probability~$1$; and 
        the \emph{limit-sure winning} asks whether the objective can be satisfied with probability arbitrarily close to~$1$.
    \item[(b)] 
        \emph{Quantitative analysis} asks to approximate the maximum or optimal probability with which the objective can be satisfied, up to a given additive error.
\end{itemize}

\paragraph{General Undecidability}
Most results for POMDPs and the computational analysis are negative (undecidability results).
Indeed, the limit-sure analysis for reachability objectives is undecidable~\cite{gimbert2010probabilistic, chatterjee2010probabilistic}, which extends to general parity objectives.
The almost-sure analysis for reachability objectives is \textnormal{EXPTIME}-complete \cite{baierDecisionProblemsProbabilistic2008, chatterjee2010qualitative}, but is undecidable for parity objectives even for two priorities (namely coBüchi objectives)~\cite{baierDecisionProblemsProbabilistic2008, chatterjeeRandomnessFree2010}. 
The quantitative analysis for reachability objectives is undecidable~\cite{madani2003undecidability}, which extends to general parity objectives.
Given this wide range of results about non-existence of algorithms a natural direction to explore is the existence of subclasses for which the computational problems are decidable.

\paragraph{Revealing POMDPs} 
We study \emph{revealing POMDPs}, a special subclass of POMDPs where each visited state is announced to the controller with a positive probability. 
Consequently, the controller's uncertainty about the state, i.e. the probability distribution over states (the \emph{belief}), occasionally collapses into a Dirac distribution on the announced state.
Revealing POMDPs have been previously studied in several works, including \cite{chenIntermittentlyObservableMarkov2023, bellyRevelationsDecidableClass2025, avrachenkovConstrainedAverageRewardIntermittently2025}, which present motivation and applications for this model.

\paragraph{Previous Results and Open Questions}
A key result for revealing POMDPs shows that the almost-sure analysis for parity objectives is \textnormal{EXPTIME}-complete \cite{bellyRevelationsDecidableClass2025}. 
However, the limit-sure and quantitative analysis problems for reachability and parity objectives remained open for revealing POMDPs.

\paragraph{Our Contributions}
We address the open questions for revealing POMDPs. 
Our main contributions are the following.
For revealing POMDPs with parity objectives, we show that

\begin{itemize}
    \item 
        The limit-sure analysis coincides with almost-sure analysis, and consequently is \textnormal{EXPTIME}-complete.
        
    \item
        The quantitative analysis can be achieved in \textnormal{EXPTIME}, i.e., in the same complexity as for qualitative analysis.
\end{itemize}
The results for POMDPs and revealing POMDPs are summarized in 
\Cref{Table: Summary of general results} and \Cref{Table: Summary of revealing results}, respectively. 

\begin{table}[t]
 	\centering
 	\begin{tabular}{|l|c|c|}
 		\hline
 		\multicolumn{1}{|c|}{\multirow{2}{*}{Problems}} & \multicolumn{2}{c|}{Objectives} \\
 		\cline{2-3}
 		& Reachability & Parity \\
 		\hline\hline
 		Almost-sure & EXP-complete & Undecidable \\
 		\hline
 		Limit-sure & Undecidable & Undecidable \\
 		\hline
 		Quantitative & Undecidable & Undecidable \\
 		\hline
 	\end{tabular}
    \caption{
 		Computational complexity for POMDPs with reachability and parity objectives.
 	}
 	\label{Table: Summary of general results}
\end{table}

\begin{table}[t]
 	\centering
 	\begin{tabular}{|l|c|c|}
 		\hline
 		\multicolumn{1}{|c|}{\multirow{2}{*}{Problems}} & \multicolumn{2}{c|}{Objectives} \\
 		\cline{2-3}
 		& Reachability & Parity\\
 		\hline\hline
 		Almost-sure & EXP & EXP-complete \\
 		\hline
 		Limit-sure & \textbf{EXP} & \textbf{EXP-complete} \\
 		\hline
 		Quantitative & \textbf{EXP} & \textbf{EXP} \\
 		\hline
 	\end{tabular}
     \caption{
 		Computational complexity for revealing POMDPs with reachability and  parity objectives.
        Our contributions are marked in bold.
 	}
 	\label{Table: Summary of revealing results}
\end{table}

\paragraph{Technical Contributions}
A closely related work established that almost-sure analysis for parity objectives in revealing POMDPs is \textnormal{EXPTIME}-complete \cite{bellyRevelationsDecidableClass2025}.
Additionally, \cite{chatterjeeWhatDecidablePartially2016} previously showed that almost-sure analysis for parity objectives in general POMDPs under finite-memory policies is \textnormal{EXPTIME}-complete. 
Therefore, the \textnormal{EXPTIME}-completeness result naturally follows once finite memory policies are proven sufficient for almost-sure analysis.
However, limit-sure analysis and quantitative analysis for POMDPs remain undecidable in general, even under finite memory policies~\cite{chatterjeeFiniteMemoryStrategiesPOMDPs2021}. 
This highlights a sharp contrast with the almost-sure analysis and gives rise to new technical challenges. 

First, we consider revealing POMDPs with the belief-reachability objectives. 
Given a set of target states, the belief-reachability objectives consider the probability of eventually observing such states. 
We prove that the quantitative analysis of belief-reachability in revealing POMDPs  can be achieved in \textnormal{EXPTIME}. 
The argument proceeds in two steps: $(i)$ we prove that this objective can be approximated by their finite-horizon counterparts; and $(ii)$ we show that this finite-horizon objective can be approximated by a point-based approximation.

Then, we consider revealing POMDPs with the parity objectives. 
We show that the value for parity objectives corresponds to the value for belief-reachability objectives to almost-sure winning states.
Finally, we prove the existence of optimal policies for parity objectives in revealing POMDPs. 

Proofs and details omitted due to space restrictions are provided in the Appendix.

\paragraph{Related Works}
The study of subclasses of POMDPs with tractable algorithms is a broad topic with many directions.
For example, subclasses of POMDPs for qualitative analysis have been explored in~\cite{fijalkow2015deciding,chatterjee2012decidable};
and for quantitative analysis various subclasses have also been studied such as with ergodicity condition~\cite{chatterjee2024ergodic}; multiple environments only~\cite{van2023robust,chatterjee2025value}; or in online learning~\cite{liu2022partially,chen2023lower}. 
Our work focuses on such a subclass, namely revealing POMDPs, which has been studied in the literature.

\citet{avrachenkovConstrainedAverageRewardIntermittently2025} studied strongly-connected revealing POMDPs and restricting attention to the set of belief-stationary policies.
They proved that there is an optimal policy within this set, whose induced belief dynamics are contracting and reach a positive recurrent class of beliefs in finite time.
They also provided an infinite-dimensional linear program that characterizes the optimal belief-stationary policy. 
They considered an application of strongly-connected revealing blind MDPs.
\cite{chenIntermittentlyObservableMarkov2023} introduced the model of revealing blind MDPs.
They consider the discounted objective and present algorithms based on finite-horizon approximation to compute the discounted value.
Our work addresses the open computational analysis questions for revealing POMDPs.


\section{Preliminaries}
\label{Section: Preliminaries}

\paragraph{Notation}
For a positive integer $n$ the set $\{1, 2, \ldots, n\}$ is denoted by $[n]$.
Sets and correspondences are denoted by calligraphic letters, e.g., $\S, \A, \Z$. 
Elements of these sets are denoted by lowercase letters, e.g., $s, a, z$.
Random elements with values in these sets are denoted by uppercase letters, e.g., $S, A, Z$. 
The set of probability measures over a finite set $\S$ is denoted by $\Delta(\S)$.
The Dirac measure at some element $s \in \S$ is denoted by $\1[s]$.
The support of a probability measure $b \in \Delta(\S)$ is denoted by $\supp(b)$.

\begin{definition}[POMDP]   
    A POMDP is a tuple $\POMDP = (\S, \A, \Z,\transition, b_0)$, where:
    \begin{itemize}
        \item 
            $\S$ is a finite set of states;
        \item
            $\A$ is a finite set of actions;
        \item
            $\Z$ is a finite set of signals;
        \item
            $\transition \colon \S \times \A \to \Delta(\S \times \Z)$ is a probabilistic transition function that, given a state $s$ and an action $a$, returns the probability distribution over the successor states and signal;
        \item
            $b_0 \in \Delta(\S)$ is the initial belief over the states.
    \end{itemize}
\end{definition}
\noindent Markov Decision Processes (MDPs) are POMDPs in which the signal corresponds to the state~\cite{puterman1994}. 
Formally, an MDP is denoted by $M = (\S, \A, \transition, b_0)$ with transition function $\transition \colon \S \times \A \to \Delta(\S)$. 

\paragraph{Dynamic}
Given a POMDP, a controller knows all defining parameters. 
At the beginning, nature draws a state $S_0 \sim b_0$, which is not informed to the controller. 
Then, at each step $t \ge 0$, the controller chooses an action $A_t \in \A$, possibly at random.
In response, nature draws the next state and signal $(S_{t + 1}, Z_{t + 1}) \sim \transition(S_t, A_t)$.
The signal $Z_{t + 1}$ is revealed to the controller, but the state $S_{t + 1}$ is not announced.
A play (or a path) in the POMDP is an infinite sequence $\play = (s_0, a_0, z_1, s_1, a_1, z_2, s_2, a_2, \ldots)$ of states, actions, and signals such that, for all $t \ge 0$, we have $\transition(s_t, a_t)(s_{t + 1}, z_{t + 1}) > 0$. 
The set of all plays is denoted by $\Plays$.

\paragraph{Belief}
At each step $t \ge 0$, the controller's (random) belief about the current state can be computed using Bayes' rule and is denoted by $B_t \in \Delta(\S)$. 
In the cases where the controller knows the exact state the controller's belief is a Dirac measure $\1[s]$ at some state $s \in \S$.

\begin{definition}[Revealing POMDP]
    A POMDP is revealing if, each time a state is visited, the state is also announced to the controller with positive probability.
    Formally, for each state there is a designated signal, i.e., $\S \subseteq \Z$, and, for all states $s, \next{s} \in \S$ and actions $a \in \A$,
    \begin{align*}
        &\sum_{z \in \Z} \transition(s, a)(\next{s}, z) > 0\\
        &\qquad \implies \sum_{\another{s} \in \S} \transition(s, a)(\another{s}, \next{s}) = \transition(s, a)(\next{s}, \next{s}) > 0 \,.
    \end{align*} 
\end{definition}

Revealing POMDPs coincide with the class of strongly revealing defined in~\cite{bellyRevelationsDecidableClass2025}.
Indeed, if the controller observes a signal $s \in \Z$, then their next belief is $\1[s]$, which corresponds to a revelation of the state.

\begin{remark}[Signals]
    The signaling structure of POMDPs is modeled in different ways in the literature. 
    We comment on two general cases.
    \begin{itemize}
        \item 
            The controller may receive a set of signals at each step with a transition function $\transition \colon \S \times \A \to \Delta(\S \times 2^{\Z})$.
            In that case, we consider each set of signals as one signal $\another{\Z} \defas 2^{\Z}$ and reduce to our model.
            Even with exponentially many signals encoded in a polynomial size input, our \textnormal{EXPTIME} upper bound holds.        
        \item 
            In general POMDPs, it is enough to consider transition functions that pair each state with only one signal, known as deterministic observation, see~\cite[Remark 4]{chatterjeeWhatDecidablePartially2016}.
    For revealing POMDPs, deterministic observations do not capture the general case because they correspond to MDPs.
    \end{itemize}
\end{remark}

\paragraph{Policies}
A \emph{policy} defines how a controller selects actions based on all the information available up to a given step.
Formally, a (history-dependent randomized) policy is a function $\policy \colon (\A \times \Z)^* \to \Delta(\A)$. 
The set of all policies is denoted by $\Policies$.
A policy is pure if it prescribes deterministic actions, i.e., it corresponds to a function $\policy \colon (\A \times \Z)^* \to \A$.

\paragraph{Probability Measures}
For a finite prefix of a play, an element in $(\S \times \A)^*$, its cone is the set of plays with it as their prefix. 
Given a policy $\policy$ and an initial belief $b_0$, the unique probability measure over Borel sets of infinite plays obtained given $\sigma$ is denoted by $\PP^\policy_{b_0}( \cdot )$, which is defined by Carathéodory's extension theorem by extending the natural definition over cones of plays~\cite{billingsley2012ProbabilityMeasurea}. 

\paragraph{Objectives}
An objective in a POMDP is a Borel set of plays $\Phi \subseteq \Plays$ in the Cantor topology on $\Plays$~\cite{kechrisClassicalDescriptiveSet1995}.
We consider objectives in the first $2\sfrac{1}{2}$ levels of the Borel hierarchy, including the parity objective which expresses all $\omega$-regular objectives~\cite{thomasLanguagesAutomataLogic1997}.
Denote the state at time $t$ by $s_t$ and set of states that occur infinitely often in a play $\play$ by $\I(\play) \defas \{ s \in \S : \forall t \ge 0 \; \exists \another{t} \ge t \quad s = s_{\another{t}} \in \play \}$.
We consider the following objectives.
\begin{itemize}
	\item 
		\emph{Reachability:}
		Given a set $\X \subseteq \S$ of target states, the reachability objective requires that a target state is visited at least once.
		Formally, the reachability objective is 
		\[
		\Reach(\X) \defas \{ \play \in \Plays : \exists t \ge 0 \quad s_t \in \X \}.
		\]
	\item 
		\emph{Parity:}
		Given a $d \ge 0$ and a function $\priority \colon \S \to \{0, 1, \ldots, d \}$ of priorities, the parity objective requires that the smallest priority that appears infinitely often is even.
		Formally, 
		\[
		\Parity \defas  \left\{ \play \in \Plays : \min \{ \priority(s) : s \in \I(\play) \} \text{ is even} \right\}.
		\]
\end{itemize}

\paragraph{Value}
The value of an objective in a POMDP is the maximum probability a controller can guarantee to satisfy the objective. 
Formally, given an objective $\Phi \subseteq \Plays$, the value is a function of the initial belief $\val_\Phi \colon \Delta(\S) \to [0, 1]$ defined by $\val_{\Phi}(b) \defas \sup_{\policy\in \Policies} \PP^{\policy}_{b}( \play \in \Phi)$.
Denote the reachability and parity values by $\val_{\textnormal{R}(\X)}$ and $\val_{\textnormal{P}}$, respectively. 
We may omit $\X$ in $\val_{\textnormal{R}(\X)}$ if it is clear from the context.

\paragraph{Approximately Optimal Policies}
Given $\eps \ge 0$ and an objective $\Phi \subseteq \Plays$, a policy $\policy \in \Policies$ is $\eps$-optimal if it guarantees the value up to an additive error $\eps$, i.e., if $\PP_b^{\policy}(\play \in \Phi) \ge \val_{\Phi}(b) - \eps$.
In particular, we call $0$-optimal policies simply optimal.

\paragraph{Computational Analysis}
The computational analysis problems for POMDPs with an objective $\Phi$ are: 
\begin{itemize}
\item  \emph{Almost-sure} analysis asks whether the objective can be satisfied with probability~$1$, i.e., $\exists \policy\in \Policies$ such that $\PP_{b_0}^{\policy}(\play \in \Phi) = 1$.
\item \emph{Limit-sure} analysis asks whether the objective can be satisfied with probability arbitrarily close to~1, i.e.,
 $\forall \eps > 0$ $\exists \policy\in \Policies$ such that $\PP_{b_0}^{\policy}(\play \in \Phi) \ge 1 - \eps$ or equivalently whether
 $\val_{\Phi}(b_0)=1$.
\item \emph{Quantitative} analysis asks to compute an approximation of the optimal value, i.e., for all $\eps > 0$, provide $v \in [0, 1]$ such that $\left|v - \val_{\Phi}(b_0)\right| \le \eps$.
\end{itemize}


\section{Overview}
\label{Section: Overview}
We present an overview of our approach and the results.

\subsection{Overview of Approach}
To solve the qualitative and quantitative analysis of parity objectives, we introduce a new objective called belief-reachability and show that the parity value coincides with the belief-reachability value to the set of almost-sure winning parity states. 

\paragraph{Belief-Reachability} 
Given a set of target states $\X \subseteq \S$, the belief-reachability objective requires that a target state is visited and the controller has complete knowledge of the state at that step at least once.
Formally, the set of Dirac beliefs on $\X\subseteq \S$ is denoted by $\D_\X \defas \left\{ b \in \Delta(\S) : \exists s \in \X \textnormal{ such that } b = \1[s] \right\}$.
Then, the belief-reachability objective is $\BeliefReach(\X) \defas \{ \play \in \Plays : \exists t \ge 0 \quad B_t(\play) \in \D_{\X} \}$. 
Denote the belief-reachability value by $\val_{\textnormal{BR}(\X)}$. 
We may omit $\X$ in $\val_{\textnormal{BR}(\X)}$ if it is clear from the context.

\begin{remark}[Generality of belief-reachability]
\label{Result: belief-reachabilty covers reachability}
    In general POMDPs with reachability objectives, without loss of generality, target states can be considered absorbing, i.e., the dynamic remains in the same state once a target state is visited.
    Furthermore, one can reduce to the case where there is only one target state.
    These simplifications affect the dynamic, but neither optimal policies nor the reachability value.
    Belief-reachability generalizes reachability objectives as follows.
    Given a POMDP $\POMDP$ with an absorbing target state $s^*$, consider a copy of $\POMDP$ but add a new action and two absorbing states $\top$ and $\bot$.
    After playing the new action, the state moves to $\top$ if the state was in $s^*$, and to $\bot$ otherwise.
    Moreover, the controller is announced which of these states is reached.
    Then, the belief-reachability value to $\top$ coincides with the original reachability value.
    Indeed, for an approximately optimal policy of $\POMDP$, play the new action after sufficiently many steps to obtain approximately the same value in the belief-reachability objective.
\end{remark}

\subsection{Overview of Results}
\paragraph{Results for Belief-Reachability Objectives}
Our main results for belief-reachability objectives are the following, which are proved in \Cref{Section: Reachability}.

\begin{restatable}{theorem}{approxbeliefreach}
\label{Result: Approximating belief-reachability is EXPTIME}
    Quantitative analysis for belief-reachability objectives for revealing POMDPs is in \textnormal{EXPTIME}.
\end{restatable}

By \Cref{Result: belief-reachabilty covers reachability}, reachability objectives reduce to belief-reachability objectives. 
Therefore, they have the following consequence.

\begin{corollary}
\label{Result: reachability is EXPTIME}
    Quantitative analysis for reachability objectives for revealing POMDPs is in \textnormal{EXPTIME}.
\end{corollary}

\paragraph{Results for Parity Objectives}
Our main results for parity objectives are the following, which are proved in \Cref{Section: Parity}.

\begin{restatable}{theorem}{approxparity}
\label{Result: approximating of parity value is exptime}
    Quantitative analysis for parity objectives for revealing POMDPs is in \textnormal{EXPTIME}.
\end{restatable}

\Cref{Result: approximating of parity value is exptime} generalizes \Cref{Result: reachability is EXPTIME} to parity objectives, and follows from a reduction of parity objectives to belief-reachability to a set that can be computed in \textnormal{EXPTIME}, see \Cref{Result: Reduction of parity to belief-reachability}.

\begin{restatable}{theorem}{parityoptimal}
\label{Result: Revealing POMDPs with parity objectives have 0-optimal policy}
    Optimal policies exist for parity objectives for revealing POMDPs.
\end{restatable}

The following results follow from \cite[Theorem 3]{bellyRevelationsDecidableClass2025} and \Cref{Result: Revealing POMDPs with parity objectives have 0-optimal policy}.

\begin{restatable}{corollary}{limitsure}
\label{Result: limit-sure-1}
    For revealing POMDPs with parity objectives, limit-sure and almost-sure winning coincide, and limit-sure analysis is in \textnormal{EXPTIME-complete}.
\end{restatable}


\section{Belief-Reachability Objectives}
\label{Section: Reachability}

In this section, we prove \Cref{Result: Approximating belief-reachability is EXPTIME}, i.e., that the quantitative analysis for belief-reachability objectives is in \textnormal{EXPTIME}. 
We proceed in five steps:
\begin{itemize}
    \item 
        We introduce reliable actions, which preserve the current belief-reachability value. 
        \Cref{Result: Every belief has a reliable action} proves that every belief has a reliable action.
    \item 
        \Cref{Result: Stopping and reliable is optimal} shows that stopping policies that play only reliable actions are optimal.
    \item 
        \Cref{Result: Existence of optimal policy} shows that playing reliable actions uniformly at random is stopping.
    \item 
        \Cref{Result: Stopping optimal policy approximate the value} upper bounds the horizon needed to approximate the belief-reachability value, using stopping optimal policies.
    \item 
        \Cref{Result: Approximation of the T step reachability value} presents a point-based dynamic programming algorithm to approximate the finite-horizon belief-reachability value. 
\end{itemize}

\begin{definition}[Reliable Action]
    Consider a POMDP with belief-reachability objectives. 
    An action $a\in \A$ is \emph{reliable} if it preserves the belief-reachability value.
    Formally, for each belief $b \in \Delta(\S)$, define the set of reliable actions $\R(b)$ by
    \[
       \R(b) \defas \left\{ a \in \A : \EE^a_{b}(\val_{\textnormal{BR}(\X)}(B_1)) = \val_{\textnormal{BR}(\X)}(b) \right\} \,.
    \]
\end{definition}

Because the set of actions $\A$ is finite, we have the following result.

\begin{restatable}{proposition}{existencereliable}
\label{Result: Every belief has a reliable action}
    The set of reliable actions is nonempty.
\end{restatable}

\begin{definition}[Terminal State]
    Consider a POMDP and a set of target states $\X \subseteq \S$. 
    A state $s \in \S$ is called \emph{terminal} (for $\X$) if $s \in \X$
    or if $\X$ cannot be reached from $s$ with positive probability, i.e., $\sup_{\policy\in \Policies} \PP_{\1[s]}^{\policy}(\exists t \ge 0 \quad \supp(B_t) \cap \X \not = \emptyset ) = 0$. 
    The set of terminal states is denoted by $\T$.
\end{definition}

\begin{definition}[Stopping Policy]
    Consider a POMDP, a set of target states $\X \subseteq \S$, and a corresponding set of terminal states $\T \subseteq \S$. 
    For $n \in \NN$ and $q > 0$, a policy $\policy\in \Policies$ is called $(n, q)$-\emph{stopping} (for $\X$) if within $n$ steps the belief collapses to a Dirac mass on a terminal state with probability at least $q$.
    Formally, from every initial belief $b \in \Delta(\S)$,
    \[
        \PP_b^\policy( \exists t \le n \quad B_t \in \D_{\T}) \ge q \,.
    \]
    We say that a policy is stopping if it is stopping for some $n$ and $q > 0$.
\end{definition}

Similar to \cite[Lemma 5]{anderssonComplexitySolvingStochastic2009} that considers finite duration games, it is easy to see that, if a policy is stopping and uses only reliable actions, then it is optimal.

\begin{restatable}{lemma}{stoppingandreliable}
\label{Result: Stopping and reliable is optimal}
    Consider a POMDP with belief-reachability objectives.
    If a policy is stopping and uses only reliable actions, then it is optimal.
\end{restatable}

\begin{proof}[Proof Sketch]
    Consider a stopping policy $\policy\in \Policies$ which uses only reliable actions. 
    Because every action chosen by $\policy$ is reliable, the belief‑reachability values $\val_{\textnormal{BR}}(B_t)$ forms a martingale, i.e., $\val_{\textnormal{BR}}(b_0) = \EE_{b_0}^{\policy} \bigl(\val_{\textnormal{BR}}(B_t)\bigr)$. 
    Moreover, since the policy $\policy$ is stopping, it reaches a Dirac belief on a terminal state in finite time with probability~1 and the hitting time $\tau$ of $\D_\T$ is almost-surely finite. 
    Therefore, we have $\val_{\textnormal{BR}}(b_0) = \EE_{b_0}^{\policy} \bigl(\val_{\textnormal{BR}}(B_\tau)\bigr)$.
\end{proof}

\paragraph{Minimal Positive Transition}
Denote $\transition_{\min}$ the minimum non-zero probability in the transition function, i.e., 
\begin{align*}
    \transition_{\min}\defas 
        \min \Bigl \{
                \transition(s, a)(\next{s}, z)
                \;:\;
                \forall s,\, \next{s} \in \S,\;
                a \in \A,\;
                z \in \Z \quad
                \transition(s, a)(\next{s}, z) > 0
            \Bigr \}.
\end{align*}

\begin{restatable}{lemma}{optimalpolicyreach}
\label{Result: Existence of optimal policy}
    Consider a revealing POMDP with belief-reachability objectives.
    The policy $\policy$ that plays reliable actions uniformly at random is $(n, q)$-stopping with parameters $n \defas |\S| + 2$ and $q \defas \transition_{\min}^2 \left( \transition_{\min}/|\A| \right)^{|\S|}$.
    
\end{restatable}

\begin{proof}[Proof Sketch]
    By \Cref{Result: Every belief has a reliable action}, every belief has a reliable action, so the policy $\policy$ is well-defined. 
    Build the layered sets $\S_0 \defas T$, $\S_{t+1} \defas \{s : \exists a \in \R(\1[s]) \text{ with a positive‑probability move to } \S_t\}$. 
    The sequence covers all states after at most $|\S|$ iterations, i.e., $\S_{|\S|}=\S$, because, if there were states outside, then they require to use unreliable actions to get to terminal states, which lowers their value and forms a contradiction. 
    Under the policy $\policy$ that plays every reliable action uniformly, (a)~the first ``reveal'' of the Dirac belief happens with probability $\transition_{\min}$; (b)~each step toward the next layer then succeeds with probability at least $\transition_{\min}/|\A|$; and (c)~the final ``reveal'' of the Dirac belief happens with probability $\transition_{\min}$. 
    Hence, within $n \defas |\S|+2$ steps, a Dirac belief is reached on a terminal state with probability $q \defas \transition_{\min}^2 \left(\transition_{\min}/|\A| \right)^{|\S|}$, which proves that $\policy$ is $(n,q)$-stopping.
\end{proof}

We deduce directly from \Cref{Result: Stopping and reliable is optimal} and \Cref{Result: Existence of optimal policy} the next result.

\begin{corollary}
\label{Result: Optimal policy for reachability}
    Every revealing POMDP with belief-reachability objectives has an optimal policy.
\end{corollary}

We turn to the quantitative analysis for the belief-reachability value.
Note that the existence of optimal stopping policies implies that this value can be approximated using a finite horizon.

\begin{restatable}{lemma}{stoppingapproxreach}
\label{Result: Stopping optimal policy approximate the value}
    Consider a POMDP with belief-reachability objectives and an $(n, q)$-stopping optimal policy $\policy$.
    For all $\eps > 0$, the finite horizon $T \defas n \left\lceil \frac{\log(\eps)}{\log(1-q)} \right\rceil$ is such that
    \[
        \PP^{\policy}_{b_0} ( \exists t \le T \quad B_t \in \D_\T )
            \ge \val_{\textnormal{BR}}(b_0) - \eps \,.
    \]
\end{restatable}

\begin{proof}[Proof Sketch]
Split the play into blocks of $n$ steps.  
Because the optimal policy $\policy$ is $(n,q)$-stopping, each block reaches a Dirac belief on a terminal state with probability at least $q$. Thus for the hitting time $\tau$ of $\D_\T$ can be bounded by
$\PP_{b_0}^{\policy}( \tau > kn) \le (1-q)^{k}$. Therefore, $\tau$ has a geometric tail and is almost surely finite.  
Choose $T = n \left\lceil \dfrac{\log(\eps)}{\log(1-q)} \right\rceil$ so that $\PP^{\policy}_{b_0} ( \tau > T) \le \eps$, showing that a horizon of $T$ steps suffices to approximate the belief-reachability value.
\end{proof}

\paragraph{Reduction to Finite Horizon}
To solve the quantitative analysis for revealing POMDPs with belief-reachability objectives, it suffices to compute the finite horizon value for a horizon that is exponentially large on the input size. 
Indeed, by \Cref{Result: Existence of optimal policy}, there exists an $(n, q)$-stopping optimal policy with $n=|\S|+2$ and $q = (\transition_{\min})^2(\transition_{\min}/|\A|)^{|\S|}$. 
Hence, by \Cref{Result: Stopping optimal policy approximate the value}, for every $\eps > 0$, we can consider $T = n\left\lceil\frac{\log(\eps)}{\log(1-q)}\right\rceil = O\left( |\S| |\A|^{|\S|} \log(1/\eps) / \transition_{\min}^{|\S| + 2} \right)$, which is at most exponentially large in the input size.

\paragraph{Previous Approaches to Finite Horizon}
Finite horizon objectives have been studied by a long time.
A naive approach to the quantitative analysis for belief-reachability objectives would take the time bound $T$ in \Cref{Result: Stopping optimal policy approximate the value} and solve a (fully observable) MDP with $|\A \times \S|^{T}$ states.
This approach implies a 2\textnormal{EXPTIME} complexity upper bound. 
A fundamental result states that, for horizons that are polynomially large with respect to the input, computing the finite horizon value of POMDPs is PSPACE-complete \cite[Theorem 6, page 448]{papadimitriouComplexityMarkovDecision1987}.
The technique, instead of listing explicitly all exponentially many histories, compactly represents them using nondeterminism and space proportional to the horizon.
The conclusion follows from \textnormal{PSPACE} being closed under nondeterminism by Savitch's theorem. 
For exponentially large horizons, this technique leads to an EXPSPACE upper bound.
Instead, we prove an \textnormal{EXPTIME} upper bound.

Point-based algorithms were introduced for POMDPs by~\cite{pineauPointbasedValueIteration2003} as an alternative to approximating the value of POMDPs by considering a fixed subset of beliefs and projected belief updates.
For a survey on point-based algorithms see \cite{shaniSurveyPointbasedPOMDP2013, walravenPointBasedValueIteration2019}, which focuses on finite horizon objectives but does not provide the complexity upper bound we require.

\paragraph{Finite-Horizon Belief-Reachability}
Let $T\in \NN$ be a finite horizon.
The $T$-step belief-reachability value to $\X \subseteq \S$ is defined by $\val_{\textnormal{BR}(\X)}^T(b)\defas\max_{\policy\in \Policies}\PP_b^\policy(\exists t \le T \;\, B_t\in \D_\X).$

\begin{restatable}{lemma}{tstepapprox}
\label{Result: Approximation of the T step reachability value}
    Approximating the $T$-step belief-reachability value of revealing POMDPs up to an additive error of $\eps > 0$ can be computed in exponential time, formally, in time $O\left( T^{|\S|} |\S|^{|\S|+1} |\A| |\Z| / \eps^{(|\S| - 1)} \right)$.
\end{restatable}

\begin{proof}[Proof Sketch]
    We approximate the $T$-step belief-reachability value by discretizing the belief simplex and applying projected Bellman updates over this grid.
    Define a $k$-uniform grid $G_k \subseteq \Delta(\S)$ so that every belief $b \in \Delta(\S)$ is at $L_1$-distance at most $\eps/(T+1)$ from some grid point $\Pi_k(b)$, where
    $k = \left \lceil (T + 1)|\S| / \eps \right \rceil$.
    Then, compute the Bellman updates on grid points, carefully projecting back to $G_k$ after each update, for $T$ steps. 
    By induction on the horizon, the error at each step accumulates by at most $\eps/(T+1)$, yielding a total error at most $\eps$. 
    The grid has size $O((T|\S|/\eps)^{|\S|-1})$, and each Bellman update takes $O(|\S|^2|\A||\Z|)$ time, so the total running time is
    $O\left(T^{|\S|}|\S|^{|\S|+1}|\A||\Z|\eps^{-(|\S|-1)}\right)
    $, which is exponential in the input size.
\end{proof}

The following result follows from \Cref{Result: Stopping optimal policy approximate the value,Result: Approximation of the T step reachability value}.
\approxbeliefreach*


\section{Parity Objectives}
\label{Section: Parity}

\subsection{Reduction to Belief-Reachability Objectives}

In the sequel, we show that, for revealing POMDPs, the parity objective value coincides with the belief-reachability to the set of Dirac on the almost‐sure winning parity states.
This proof uses the standard notion of \emph{end components} in MDPs, introduced in \cite{de1997formal}; see also \cite[Section 10.6.3]{baier2008principles}.

\paragraph{Underlying MDPs of POMDPs} 
For a POMDP $P = (\S, \A, \Z, \transition, b_0)$, we define its underlying MDP as $M = (\S \times (\Z \cup \{ \square \}), \A, \another{\transition}, \another{b}_0)$, where the transition function is defined as, 
for all $s, \next{s} \in \S, z \in \Z \cup \{ \square \}, \next{z} \in \Z$, 
\[
    \another{\transition} \bigl((s, z), a \bigr) \bigl((\next{s}, \next{z}) \bigr) \defas \transition(s, a)(\next{s}, \next{z})\,, \qquad 
\]
and the initial belief is defined as $\another{b}_0 \bigl( (s, \square) \bigr) \defas b_0(s)$ for all states $s \in \S$. This construction captures the entire dynamic of $P$. We use this notion to analyze the  policies derived from the original POMDP.

\paragraph{End Components} 
Let $M = (\S, \A, \transition, b_0)$ be an MDP. 
An end component is a pair $\U = (\Q, \E)$ where $\Q \subseteq \S$ is a subset of states and $\E \colon \Q \rightrightarrows \A$ assigns each state to a set of nonempty actions such that
\begin{itemize}
    \item 
        \emph{Closedness}. 
        For all states $q \in \Q$ and actions $a \in \E(q)$, we have $\supp(\transition(q, a)) \subseteq \Q$; and
    \item 
        \emph{Strong connectivity}.
        The directed graph with states $\Q$ and edges $(q, q^\prime)$ where there exists $a \in \E(q)$ such that $\transition(q, a)(q^\prime) > 0$ is strongly connected. 
\end{itemize}
A play $\play$ visits infinitely often an end component $\U = (\Q, \E)$, denoted by $\I(\play) = (\Q, \E)$, if 
\begin{itemize}
    \item The set of states visited infinitely often in $\play$ is $\Q$; and 
    \item For all state $q \in \Q$, the set of actions played infinitely often when in $q$ is $\E(q)$.
\end{itemize}
We say $\I(\play) \subseteq (\Q, \E)$ if $\I(\play) = (\another{\Q}, \another{\E})$, $\another{\Q} \subseteq \Q$ and~$\another{\E} \subseteq \E$.

The following statement is a fundamental result of end components in MDPs.
\begin{lemma}[\protect{\cite[Theorems 3.1 and 3.2]{de1997formal}}]
\label{Result: End components}
    For MDPs, the following assertions hold.
    \begin{itemize}
        \item 
            For every end component $\U = (\Q, \E)$ and state $q \in \Q$, there exists a policy $\policy$ such that 
            \[
                \PP_{\1[q]}^{\policy} \left( \I(\play) = \U \right) = 1 \,.
            \]
        \item 
            For all policies $\policy$, we have 
            \[
                \PP_{b_0}^{\policy} \left( \I(\play) \textnormal{ is an end component} \right) = 1 \,.
            \]
    \end{itemize}
\end{lemma}

The following result relates end components in MDPs to policies in general POMDPs. 

\begin{restatable}{lemma}{almostsuresafetygeneral}
\label{Result: Existence of almost-sure saftey policy for an end component in general PODMPs}
    Consider a POMDP $\POMDP$ and its underlying MDP $\MDP$.
    For every reachable end component $\U = (\Q, \E)$ of $\MDP$, i.e., there exists a policy $\policy$ on $\POMDP$ such that $\PP_{b_0}^{\policy} ( \I(\play) = \U ) > 0$, we have that, for all states $q \in \Q$, actions $a \in \E(q)$, and signals $z \in \Z$, there exists a policy~$\policy_{q,a,z}$ on $\POMDP$ such that
    \[
        \PP_{b}^{\policy_{q,a,z}} ( \I(\play) \subseteq \U) \;=\; 1 \,,
    \]
    where $b$ is the belief after starting with belief $\1[q]$, playing action $a$, and receiving signal $z$.
\end{restatable}

\begin{proof}[Proof Sketch]
    Consider a reachable end component $\U=(\Q,\E)$ of $\MDP$ with corresponding policy $\policy$, i.e.,
    $\PP^{\policy}_{b_0}(\I(\play)=\U) > 0$. 
    On the event $\{ \I(\play)= \U \}$, every state-action pair $(q,a)$ with $q \in \Q$ and $a \in \E(q)$ is taken infinitely often, and each such occurrence results in at least one signal $z$ that occurs with positive probability.
    Fix $q \in \Q$, $a \in \E(q)$ and a signal $z$ that can follow $(q,a)$. 
    By contradiction, assume that from the belief $b$ obtained after starting with $\1[q]$, playing action $a$, and receiving signal $z$, no policy can stay inside $\U$ almost surely. 
    Then, starting from $b$, any policy has a positive, state–independent strictly positive lower bound on the probability of leaving $\Q$ within a bounded number of steps.
    Because $(q, a)$ is visited infinitely often on the event $\{ \I(\play) = \U \}$, these lower‑bounded leaving probabilities accumulate, driving the probability of ever leaving $\Q$ to~$1$. 
    This contradicts $\PP^{\policy}_{b_0}(\I(\play) = \U) > 0$.
    Hence, such a leaving bound cannot exist: there must be a policy $\policy_{q,a,z}$ that, from $b$, keeps the play inside $\Q$ (and therefore inside $\U$) with probability~$1$.  
\end{proof}

The following example demonstrates that the above result is tight in the sense that it can not guarantee a policy $\policy_{q,a,z}$ such that 
$\PP_{b}^{\policy_{q,a,z}} \bigl ( \I(\play) = \U \bigr) = 1$.

\begin{example}
\label{ex:no-almost-sure-policy}
    Consider the example presented in~\cite[Example 4.3]{chatterjeeFiniteMemoryStrategiesPOMDPs2021} of a POMDP with only one possible signal in Figure~\ref{Figure: Parity example}, commonly known as a blind MDP.
    It has four states $\S = \{s_0, s_1, \bot, \top\}$ and two actions:  
    \emph{wait} ($w$) and \emph{commit} ($c$), i.e., $\A = \{w, c\}$.  
    The priority function is defined as
    \[
        \priority(s_0) \defas 1,
        \;
        \priority(s_1) \defas 1,
        \;
        \priority(\bot) \defas 1,
        \;
        \priority(\top)=0\,.
    \]
    The transitions are defined as follows.
    (a)~Under action $w$, state $s_0$ moves either to $s_1$ or loops, both with  probability $1/2$; states $s_1$ and $\bot$ loop; and state $\top$ moves to $s_0$. 
    (b) Under action $c$, state $s_0$ moves to $\bot$; state $s_1$ moves $\top$; state $\bot$ loops; and state $\top$ moves to $s_0$. 
    The only end component of the underlying MDP that satisfies the parity condition is
    \[
        \U = \bigl(\{s_1,\top\},\;\{(s_0,w),(s_1,w),(s_1,c),(\top,w),(\top,c)\} \bigr) \,.
    \]
    
    Note that, for every $\eps$‑optimal policy $\policy$, we have that $\PP_{\1[s_0]}^{\policy} ( \I(\play) = \U ) \ge 1 - \eps$.
    This is achieved only by policies that wait long enough before committing, which requires unbounded memory. 
    Conversely, no policy can guarantee $\PP_{\1[s_0]}^{\policy} (\I(\play)=\mathcal U) = 1$.
    Indeed, whenever action $c$ is taken, the belief on state $s_0$ is positive and therefore the state absorbs at $\bot$ with positive probability.
    The best that can be achieved almost surely is to stay forever in the
    end component
    \(
      \another{\U}
      =\bigl(\{s_1\},\{(s_1,w)\}\bigr),
    \)
    whose minimal priority is~1 and therefore violates the parity condition.
    \qed
\end{example}

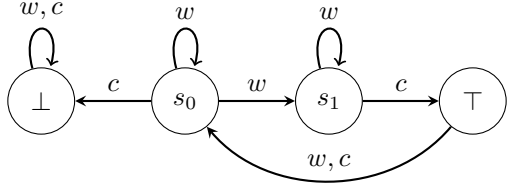
\begin{figure}[t]
    \centering
    \begin{tikzpicture} [
        node distance = 1cm, 
        every initial by arrow/.style = {thick}
        ]
        
        \node (s0) [state] {$s_0$};
        \node (s1) [state, right = of s0] {$s_1$};
        \node (bot) [state, left = of s0] {$\bot$};
        \node (top) [state, right = of s1] {$\top$};
        
        \path [-stealth, thick]
        (s0) edge [loop above] node {$w$} ()
        (s0) edge node [above] {$w$} (s1)
        (s1) edge [loop above]  node {$w$} ()
        (s0) edge node [above] {$c$} (bot)
        (s1) edge node [above] {$c$} (top)
        (bot) edge [loop above] node {$w, c$} ()
        (top) edge [bend left=50] node [above] {$w, c$} (s0);
    \end{tikzpicture}
    \caption{
        Example of a classic general POMDP that requires infinite memory policies for parity objectives.
        Edges represent a positive probability transition between states when the corresponding action in its label is used.
    }
    \label{Figure: Parity example}
\end{figure}

We now show that in revealing POMDPs, the above result can be extended to guarantee a policy $\policy_q$ such that $\PP_{\1[q]}^{\policy_q} \bigl ( \I(\play) = \U \bigr) = 1\,.$

\begin{restatable}{lemma}{almostsuresafetyrevealing}
\label{Result: Existence of almost-sure saftey policy for an end component}
    Consider a revealing POMDP $\POMDP$ and its underlying MDP $\MDP$. 
    For every policy $\policy$ on $\POMDP$ and end component $\U = (\Q, \E)$ of $M$ such that $\PP_{b_0}^{\policy} ( \I(\play) = \U ) > 0$, we have, for all states $q \in \Q$, there exists a policy~$\policy_q$ on $\POMDP$ such that
    \[
        \PP_{\1[q]}^{\policy_q} ( \I(\play) = \U) \;=\; 1 \,.
    \]
\end{restatable}

\begin{proof}[Proof Sketch]
    Consider a policy $\policy$ that $\PP^{\policy}_{b_0} (\I(\play) = \U) > 0$. 
    Fix a state $q \in \Q$. 
    We construct a policy $\policy_q$ that proceeds in \emph{phases}, one for each ordered pair $(\another{q},\another{a})$ where $\another{q} \in \Q$ and $\another{a} \in \E(\another{q})$. 
    The goal of a phase is to reach $\another{q}$ and to play action $\another{a}$. 
    When the current belief is non–Dirac, the policy follows the almost‑sure safety policy from \Cref{Result: Existence of almost-sure saftey policy for an end component in general PODMPs} that stays inside $\U$ until a Dirac belief $\1[\next{q}]$ is reached, which is guaranteed in finite time because $P$ is revealing. 
    From $\next{q}$, the policy moves inside $\U$ along a fixed finite path of ``safe'' actions to the target state $\another{q}$; if the belief becomes non‑Dirac, the policy returns to the safety policy and retries. 
    Once the belief is $\1[\another{q}]$, the policy plays the action $\another{a}$; upon observing the resulting signal $z$, switches to the safety policy and starts the next phase. 
    Because every ordered pair $(q,a) \in \E$ is visited infinitely often, each edge of $\U$ is taken infinitely often. The safety policies ensure the play never leaves $\Q$. 
    Consequently, the event $\{\I(\play)=\U \}$ is satisfied with probability~1 under $\policy_q$ when starting from the belief $\1[q]$. 
\end{proof}

We are ready to present the reduction of parity to belief-reachability of the almost-sure winning states for parity.

\begin{restatable}{lemma}{paritytobeliefreach}
\label{Result: Reduction of parity to belief-reachability}
    Consider a revealing POMDP with parity objectives.
    Denote the set of states for which, if the initial belief were a Dirac on that state, then the parity condition can be satisfied almost-surely by $\X \defas \left\{ s \in \S : \exists \policy \in \Policies \quad \PP_{\1[s]}^{\policy}( \Parity ) = 1 \right\}$. 
    Then, the parity value coincides with the belief-reachability value to $\X$, i.e., for all beliefs $b \in \Delta(\S)$,
    \[
        \val_{\textnormal{P}}(b) = \val_{\textnormal{BR}(\X)}(b) \,.
    \]
\end{restatable}

\begin{proof}[Proof sketch]
    Consider a policy $\policy\in \Policies$.
    Consider its end components in the underlying MDP on $\S \times \Z$.
    Note that each end component either does satisfy or does not satisfy the parity condition.
    Moreover, if they satisfy the parity condition, then $\policy$ is an almost-sure winning policy starting from any state inside the end component.
    Therefore, the probability of satisfying the parity condition under $\policy$ corresponds to the belief-reachability to the end components where the parity condition is satisfied. 
    In particular, in these end components, parity condition is satisfied almost-surely.
    We conclude because the policy is arbitrary.
\end{proof}

\subsection{Proofs of Main Results}
Building on \Cref{Result: Reduction of parity to belief-reachability}, which reduces parity to belief‑reachability, and the \textnormal{EXPTIME} procedure in \Cref{Result: Approximating belief-reachability is EXPTIME} for approximating the belief-reachability value, we now (a)~obtain an \textnormal{EXPTIME} algorithm for approximating parity values (\Cref{Result: approximating of parity value is exptime}); and (b)~show that limit‑sure and almost‑sure winning coincide (\Cref{Result: Revealing POMDPs with parity objectives have 0-optimal policy}).

\approxparity*

\begin{proof}[Proof Sketch]
    By \cite[Theorem 3]{bellyRevelationsDecidableClass2025}, computing almost-sure winning parity states in the revealing POMDP is \textnormal{EXPTIME}. Therefore, this result is a direct implication of \Cref{Result: Reduction of parity to belief-reachability} and \Cref{Result: Approximating belief-reachability is EXPTIME}. 
\end{proof}

\parityoptimal*

\begin{proof}[Proof Sketch]
    It is a direct implication of \Cref{Result: Optimal policy for reachability} and \Cref{Result: Reduction of parity to belief-reachability}. 
\end{proof}

\paragraph{Concluding Remarks}
In this work, we consider revealing POMDPs which have been studied in the literature and provide decidability results for the fundamental computational problems for this model. 
Interesting directions for future work include exploring the practical applicability of our algorithms and extending our decidability results to other classes of POMDPs.

%
%
\section*{Acknowledgements}
This work was supported by the ANRT under the French CIFRE Ph.D program in collaboration between NyxAir and Paris-Dauphine University (Contract: CIFRE N° 2022/0513), 
by the French Agence Nationale de la Recherche (ANR) under reference ANR-21-CE40-0020 (CONVERGENCE project),
by Austrian Science Fund (FWF) 10.55776/COE12, 
and by the ERC CoG 863818 (ForM-SMArt) grant.

\bibliographystyle{plain}
\bibliography{refs}

\begin{thebibliography}{42}
\providecommand{\natexlab}[1]{#1}

\bibitem[{Alfaro(1998)}]{de1997formal}
Alfaro, L. 1998.
\newblock \emph{{Formal Verification of Probabilistic Systems}}.
\newblock {Ph.D.} diss., Stanford University, Stanford, CA, USA.

\bibitem[{Andersson and
  Miltersen(2009)}]{anderssonComplexitySolvingStochastic2009}
Andersson, D.; and Miltersen, P.~B. 2009.
\newblock {The Complexity of Solving Stochastic Games on Graphs}.
\newblock In \emph{Algorithms and {Computation}}, volume 5878, 112–121.
  Berlin, Heidelberg: Springer.

\bibitem[{Avrachenkov, Dhiman, and
  Kavitha(2025)}]{avrachenkovConstrainedAverageRewardIntermittently2025}
Avrachenkov, K.; Dhiman, M.; and Kavitha, V. 2025.
\newblock {Constrained Average-Reward Intermittently Observable MDPs}.
\newblock arXiv:2504.13823.

\bibitem[{Baier, Bertrand, and
  Gr{\"{o}}{{\ss}}er(2008)}]{baierDecisionProblemsProbabilistic2008}
Baier, C.; Bertrand, N.; and Gr{\"{o}}{{\ss}}er, M. 2008.
\newblock {On Decision Problems for Probabilistic Büchi Automata}.
\newblock In \emph{Foundations of {Software Science} and {Computational
  Structures}}, volume 4962, 287–301. Berlin, Heidelberg: Springer.

\bibitem[{Baier, Gr{\"{o}}{{\ss}}er, and
  Bertrand(2012)}]{baier2012probabilistic}
Baier, C.; Gr{\"{o}}{{\ss}}er, M.; and Bertrand, N. 2012.
\newblock {Probabilistic $\omega$-automata}.
\newblock \emph{Journal of the ACM (JACM)}, 59(1): 1–52.

\bibitem[{Baier and Katoen(2008)}]{baier2008principles}
Baier, C.; and Katoen, J.-P. 2008.
\newblock \emph{{Principles of Model Checking}}.
\newblock Cambridge, MA, USA: MIT press.

\bibitem[{Belly et~al.(2025)Belly, Fijalkow, Gimbert, Horn, P{\'{e}}rez, and
  Vandenhove}]{bellyRevelationsDecidableClass2025}
Belly, M.; Fijalkow, N.; Gimbert, H.; Horn, F.; P{\'{e}}rez, G.~A.; and
  Vandenhove, P. 2025.
\newblock {Revelations: A Decidable Class of POMDPs with Omega-Regular
  Objectives}.
\newblock \emph{Proceedings of the AAAI Conference on Artificial Intelligence},
  39(25): 26454–26462.

\bibitem[{Bertsekas(1976)}]{bertsekas1976DynamicProgrammingStochastic}
Bertsekas, D.~P. 1976.
\newblock \emph{{Dynamic Programming and Stochastic Control}}.
\newblock New York: Academic Press.

\bibitem[{Billingsley(2012)}]{billingsley2012ProbabilityMeasurea}
Billingsley, P. 2012.
\newblock \emph{{Probability and Measure}}.
\newblock Hoboken, NJ, USA: Wiley.

\bibitem[{Bonet and Geffner(2009)}]{bonet2009solving}
Bonet, B.; and Geffner, H. 2009.
\newblock {Solving POMDPs: RTDP-Bel vs. Point-based Algorithms}.
\newblock In \emph{Proceedings of the 21st {International Joint Conference} on
  {Artificial Intelligence}}, {IJCAI}'09, 1641–1646. San Francisco, CA, USA:
  Morgan Kaufmann Publishers Inc.

\bibitem[{Chatterjee et~al.(2015)Chatterjee, Chmelik, Gupta, and
  Kanodia}]{chatterjee2015qualitative}
Chatterjee, K.; Chmelik, M.; Gupta, R.; and Kanodia, A. 2015.
\newblock {Qualitative Analysis of POMDPs with Temporal Logic Specifications
  for Robotics Applications}.
\newblock In \emph{{IEEE International Conference} on {Robotics} and
  {Automation} ({ICRA})}, 325–330. Seattle, WA, USA: IEEE.

\bibitem[{Chatterjee, Chmel{\'{i}}k, and
  Tracol(2016)}]{chatterjeeWhatDecidablePartially2016}
Chatterjee, K.; Chmel{\'{i}}k, M.; and Tracol, M. 2016.
\newblock {What is Decidable about Partially Observable Markov Decision
  Processes with $\omega$-Regular Objectives}.
\newblock \emph{Journal of Computer and System Sciences}, 82(5): 878–911.

\bibitem[{Chatterjee et~al.(2010)Chatterjee, Doyen, Gimbert, and
  Henzinger}]{chatterjeeRandomnessFree2010}
Chatterjee, K.; Doyen, L.; Gimbert, H.; and Henzinger, T.~A. 2010.
\newblock {Randomness for Free}.
\newblock In \emph{Proceedings of the 35th International Conference on
  Mathematical Foundations of Computer Science}, 246–257. Berlin, Heidelberg:
  Springer.

\bibitem[{Chatterjee, Doyen, and Henzinger(2010)}]{chatterjee2010qualitative}
Chatterjee, K.; Doyen, L.; and Henzinger, T.~A. 2010.
\newblock {Qualitative Analysis of Partially-Observable Markov Decision
  Processes}.
\newblock In \emph{International Symposium on Mathematical Foundations of
  Computer Science}, 258–269. Berlin, Heidelberg: Springer.

\bibitem[{Chatterjee et~al.(2025)Chatterjee, Doyen, Raskin, and
  Sankur}]{chatterjee2025value}
Chatterjee, K.; Doyen, L.; Raskin, J.-F.; and Sankur, O. 2025.
\newblock {The Value Problem for Multiple-Environment MDPs with Parity
  Objective}.
\newblock In \emph{52nd International Colloquium on Automata, Languages, and
  Programming (ICALP 2025)}, volume 334 of \emph{Leibniz International
  Proceedings in Informatics (LIPIcs)}, 150:1–150:17. Dagstuhl, Germany:
  Schloss Dagstuhl – Leibniz-Zentrum f{ü}r Informatik.

\bibitem[{Chatterjee and Henzinger(2010)}]{chatterjee2010probabilistic}
Chatterjee, K.; and Henzinger, T.~A. 2010.
\newblock {Probabilistic Automata on Infinite Words: Decidability and
  Undecidability Results}.
\newblock In \emph{International Symposium on Automated Technology for
  Verification and Analysis}, 1–16. Berlin, Heidelberg: Springer.

\bibitem[{Chatterjee et~al.(2024)Chatterjee, Lurie, Saona, and
  Ziliotto}]{chatterjee2024ergodic}
Chatterjee, K.; Lurie, D.; Saona, R.; and Ziliotto, B. 2024.
\newblock {Ergodic Unobservable MDPs: Decidability of Approximation}.
\newblock arXiv:2405.12583.

\bibitem[{Chatterjee, Saona, and
  Ziliotto(2021)}]{chatterjeeFiniteMemoryStrategiesPOMDPs2021}
Chatterjee, K.; Saona, R.; and Ziliotto, B. 2021.
\newblock {Finite-Memory Strategies in POMDPs with Long-Run Average
  Objectives}.
\newblock \emph{Mathematics of Operations Research}, 47(1): 100–119.

\bibitem[{Chatterjee and Tracol(2012)}]{chatterjee2012decidable}
Chatterjee, K.; and Tracol, M. 2012.
\newblock {Decidable Problems for Probabilistic Automata on Infinite Words}.
\newblock In \emph{27th {Annual IEEE Symposium} on {Logic} in {Computer
  Science}}, 185–194. Dubrovnik, Croatia: IEEE.

\bibitem[{Chen et~al.(2023)Chen, Wang, Xiong, Mei, and Bai}]{chen2023lower}
Chen, F.; Wang, H.; Xiong, C.; Mei, S.; and Bai, Y. 2023.
\newblock {Lower Bounds for Learning in Revealing POMDPs}.
\newblock In \emph{Proceedings of the 40th {International Conference} on
  {Machine Learning}}, volume 202 of \emph{{ICML}'23}, 5104–5161. Honolulu,
  Hawaii, USA: JMLR.org.

\bibitem[{Chen and Liew(2023)}]{chenIntermittentlyObservableMarkov2023}
Chen, G.; and Liew, S.-C. 2023.
\newblock {Intermittently Observable Markov Decision Processes}.
\newblock arXiv:2302.11761.

\bibitem[{{de Alfaro} et~al.(2005){de Alfaro}, Faella, Majumdar, and
  Raman}]{de2005code}
{de Alfaro}, L.; Faella, M.; Majumdar, R.; and Raman, V. 2005.
\newblock {Code Aware Resource Management}.
\newblock In \emph{Proceedings of the 5th {ACM} International Conference on
  {Embedded} Software}, {EMSOFT} '05, 191–202. New York, NY, USA: Association
  for Computing Machinery.

\bibitem[{Durbin et~al.(1998)Durbin, Eddy, Krogh, and
  Mitchison}]{durbin1998biological}
Durbin, R.; Eddy, S.~R.; Krogh, A.; and Mitchison, G. 1998.
\newblock \emph{{Biological Sequence Analysis: Probabilistic Models of Proteins
  and Nucleic Acids}}.
\newblock Cambridge, UK: Cambridge University Press.

\bibitem[{Fijalkow et~al.(2015)Fijalkow, Gimbert, Kelmendi, and
  Oualhadj}]{fijalkow2015deciding}
Fijalkow, N.; Gimbert, H.; Kelmendi, E.; and Oualhadj, Y. 2015.
\newblock {Deciding the Value 1 Problem for Probabilistic Leaktight Automata}.
\newblock \emph{Logical Methods in Computer Science}, 11.

\bibitem[{Filar and Vrieze(1997)}]{filarCompetitiveMarkovDecision1997}
Filar, J.; and Vrieze, K. 1997.
\newblock \emph{{Competitive Markov Decision Processes}}.
\newblock New York, NY, USA: Springer.

\bibitem[{Gimbert and Oualhadj(2010)}]{gimbert2010probabilistic}
Gimbert, H.; and Oualhadj, Y. 2010.
\newblock {Probabilistic Automata on Finite Words: Decidable and Undecidable
  Problems}.
\newblock In \emph{Automata, {Languages} and {Programming}}, volume 6199,
  527–538. Berlin, Heidelberg: Springer.

\bibitem[{Kaelbling, Littman, and Cassandra(1998)}]{kaelbling1998planning}
Kaelbling, L.~P.; Littman, M.~L.; and Cassandra, A.~R. 1998.
\newblock {Planning and Acting in Partially Observable Stochastic Domains}.
\newblock \emph{Artificial Intelligence}, 101(1-2): 99–134.

\bibitem[{Kaelbling, Littman, and Moore(1996)}]{kaelbling1996reinforcement}
Kaelbling, L.~P.; Littman, M.~L.; and Moore, A.~W. 1996.
\newblock {Reinforcement Learning: A Survey}.
\newblock \emph{Journal of Artificial Intelligence Research}, 4: 237–285.

\bibitem[{Kechris(1995)}]{kechrisClassicalDescriptiveSet1995}
Kechris, A.~S. 1995.
\newblock \emph{{Classical Descriptive Set Theory}}.
\newblock New York, NY, USA: Springer.

\bibitem[{Kress-Gazit, Fainekos, and Pappas(2009)}]{kress2009temporal}
Kress-Gazit, H.; Fainekos, G.~E.; and Pappas, G.~J. 2009.
\newblock {Temporal-Logic-Based Reactive Mission and Motion Planning}.
\newblock \emph{IEEE Transactions on Robotics}, 25(6): 1370–1381.

\bibitem[{Liu et~al.(2022)Liu, Chung, Szepesvari, and Jin}]{liu2022partially}
Liu, Q.; Chung, A.; Szepesvari, C.; and Jin, C. 2022.
\newblock {When Is Partially Observable Reinforcement Learning Not Scary?}
\newblock In \emph{Proceedings of Thirty Fifth Conference on Learning Theory},
  volume 178 of \emph{Proceedings of Machine Learning Research}, 5175–5220.
  PMLR.

\bibitem[{Madani, Hanks, and Condon(2003)}]{madani2003undecidability}
Madani, O.; Hanks, S.; and Condon, A. 2003.
\newblock {On the Undecidability of Probabilistic Planning and Related
  Stochastic Optimization Problems}.
\newblock \emph{Artificial Intelligence}, 147(1-2): 5–34.

\bibitem[{Papadimitriou and
  Tsitsiklis(1987)}]{papadimitriouComplexityMarkovDecision1987}
Papadimitriou, C.~H.; and Tsitsiklis, J.~N. 1987.
\newblock {The Complexity of Markov Decision Processes}.
\newblock \emph{Mathematics of Operations Research}, 12(3): 441–450.

\bibitem[{Paz(1971)}]{paz1971introduction}
Paz, A. 1971.
\newblock \emph{{Introduction to Probabilistic Automata}}.
\newblock New York, NY, USA: Academic Press.

\bibitem[{Pineau, Gordon, and Thrun(2003)}]{pineauPointbasedValueIteration2003}
Pineau, J.; Gordon, G.; and Thrun, S. 2003.
\newblock {Point-Based Value Iteration: An Anytime Algorithm for POMDPs}.
\newblock In \emph{{Proceedings of the 18th International Joint Conference on
  Artificial Intelligence}}, 1025–1030. San Francisco, CA, USA: Morgan
  Kaufmann Publishers Inc.

\bibitem[{Pogosyants, Segala, and Lynch(2000)}]{pogosyants2000verification}
Pogosyants, A.; Segala, R.; and Lynch, N. 2000.
\newblock {Verification of the Randomized Consensus Algorithm of Aspnes and
  Herlihy: a Case Study}.
\newblock \emph{Distributed Computing}, 13(3): 155–186.

\bibitem[{Puterman(2014)}]{puterman1994}
Puterman, M.~L. 2014.
\newblock \emph{{Markov Decision Processes: Discrete Stochastic Dynamic
  Programming}}.
\newblock Hoboken, NJ, USA: Wiley.

\bibitem[{Rabin(1963)}]{rabin1963probabilistic}
Rabin, M.~O. 1963.
\newblock {Probabilistic Automata}.
\newblock \emph{Information and control}, 6(3): 230–245.

\bibitem[{Shani, Pineau, and Kaplow(2013)}]{shaniSurveyPointbasedPOMDP2013}
Shani, G.; Pineau, J.; and Kaplow, R. 2013.
\newblock {A Survey of Point-Based POMDP Solvers}.
\newblock \emph{Autonomous Agents and Multi-Agent Systems}, 27(1): 1–51.

\bibitem[{Thomas(1997)}]{thomasLanguagesAutomataLogic1997}
Thomas, W. 1997.
\newblock {Languages, Automata, and Logic}.
\newblock In \emph{Handbook of Formal Languages: Volume 3 Beyond Words},
  389–455. Berlin, Heidelberg: Springer.

\bibitem[{Van Der~Vegt, Jansen, and Junges(2023)}]{van2023robust}
Van Der~Vegt, M.; Jansen, N.; and Junges, S. 2023.
\newblock {Robust Almost-Sure Reachability in Multi-Environment MDPs}.
\newblock In \emph{Tools and {Algorithms} for the {Construction} and {Analysis}
  of {Systems}}, volume 13993, 508–526. Berlin, Heidelberg: Springer.

\bibitem[{Walraven and Spaan(2019)}]{walravenPointBasedValueIteration2019}
Walraven, E.; and Spaan, M. T.~J. 2019.
\newblock {Point-Based Value Iteration for Finite-Horizon POMDPs}.
\newblock \emph{Journal of Artificial Intelligence Research}, 65: 307–341.

\end{thebibliography}

\onecolumn

\appendix


\section{Proofs of \Cref{Section: Reachability}}

\existencereliable*

\begin{proof}
    Consider a POMDP with initial belief $b_0 \in \Delta(\S)$ and set of target states $\X \subseteq \S$.
    By definition of the value, we have that $\val_{\textnormal{BR}(\X)}(b_0) = \max_{a \in \A} \EE^{a}_{b_0} (\val_{\textnormal{BR}(\X)}(B_1))$. Because the set of actions $\A$ is finite, there exists $a \in \A$ that achieves the maximum, which is a reliable action.
\end{proof}

\stoppingandreliable*
\begin{proof}
    Consider a POMDP with initial belief $b_0 \in \Delta(\S)$ and belief-reachability objectives. 
    Denote the set of target states and the corresponding set of terminal states by $\X\subseteq \S$ and $\T \subseteq \S$, respectively.
    Consider a policy $\policy\in \Policies$ that is stopping and uses only reliable actions.
    We show that $\policy$ is optimal.
    
    Denote the hitting time of terminal Dirac beliefs by $\tau \defas \inf \{ t \ge 0 : B_t \in \D_\T \}$. 
    Because the policy $\policy$ is $(n, q)$-stopping, for every belief $b \in \Delta(\S)$, we have that $\PP_b^\policy(\tau \le n) \ge q$. 
    Therefore, for every $k \ge 1$,
    \[
        \PP_b^\policy(\tau \le k n) \ge 1 - (1 - q)^k \,.
    \]
    Taking the limit as $k$ goes to infinity, we deduce that $\PP_b^\policy(\tau < \infty) = 1$.
    
    Because $\policy$ uses only reliable actions, we have that, for every $t\ge 1$,
    \[
        \val_{\textnormal{BR}(\X)}(b) 
            = \EE^{\policy}_{b} \left( \val_{\textnormal{BR}(\X)}(B_t) \right) \,.
    \]
    Therefore, for every $T \ge 1$ and $b \in \Delta(\S)$,
    \[
        \val_{\textnormal{BR}(\X)}(b) 
            = \EE^{\policy}_{b}\left( \val_{\textnormal{BR}(\X)}(B_\tau) \1_{\{\tau \le T\}} \right) + \EE^{\policy}_{b}\left( \val_{\textnormal{BR}(\X)}(B_\tau) \1_{\{\tau > T\}} \right) \,.
    \]
    Whenever $\tau \le T$, we have that there exists $t \le T$ such that $B_t\in \D_{\T}$, and thus $\val_{\textnormal{BR}(\X)}(B_t) = \1_{\{ B_t\in \D_\X \}}$.
    Taking $T$ to infinity,
    \begin{align*}
        \val_{\textnormal{BR}(\X)}(b) 
            &= \lim_{T \to \infty} \EE^{\policy}_{b}\left( \val_{\textnormal{BR}(\X)}(B_\tau) \1_{\{\tau \le T\}} \right)\\
            &= \lim_{T \to \infty} \PP^{\policy}_{b}( \exists t \le T \quad B_t \in \D_{\X})\\
            &= \PP^{\policy}_{b}(\exists t \ge 0 \quad B_t \in \D_{\X}) \,.
    \end{align*}
    Therefore, the policy $\policy$ is optimal.
\end{proof}

\optimalpolicyreach*
\begin{proof}
    Consider a revealing POMDP $\POMDP$ with belief-reachability objectives. 
    Recall that, by \Cref{Result: Every belief has a reliable action}, every belief has at least one reliable action, i.e., for every $b \in \Delta(\S)$, we have that $\R(b) \neq \emptyset$.
    Let $\policy$ be the policy that plays uniformly at random among all reliable actions.
    We show that $\policy$ is $(n, q)$-stopping, i.e., for every initial belief $b \in \Delta(\S)$,
    \[
        \PP^{\policy}_{b} (\exists t \le n \quad B_t \in \D_\T) 
            \ge q \,,
    \]
    where 
    \[
        n = |\S| + 2 \quad \textnormal{ and } \quad q = \transition_{\min}^2\left( \frac{\transition_{\min}}{|\A|} \right)^{|\S|} \,.
    \]

    Denote the set of terminal states by $\T \subseteq \S$.
    Define the set of states that can reach $\T$ using only reliable actions, indexed by the number of steps required, as follows.
    First, $\S_{0} \defas \T$ and, for $t \ge 1$,
    \[
        \S_{t} \defas \left\{ s \in \S : \exists a \in \R\left(\1[{s}]\right) \quad \sum_{\another{s} \in \S_{t - 1}, z \in \Z} \transition(s, a) \left(\another{s}, z \right) > 0 \right\} \,.
    \]
    We claim that at most $|\S|$ steps are required to cover all states, i.e., $\S_{|\S|} = \S$.
    
    Indeed, by contradiction, assume that the set of uncovered states $\L \defas \S \setminus \S_{|\S|}$ is nonempty.
    Note that $(\S_{t})_{t \ge 0}$ is an increasing and bounded sequence and therefore reaches a fixpoint equal to $\S_{|\S|}$. 
    On the one hand, for all state $s \in \L$, under reliable actions for the belief $\1[s]$, the state can not move from $\L$ to $\S_{|\S|}$.
    Formally, for all $s \in \L, a \in \R(\1[s])$,
    \[
        \supp(\transition(s, a)) \subseteq \L \times \Z \,.
    \]
    On the other hand, using unreliable actions imply decreasing the future value, i.e., for all $s \in \L$,
    \[
        \max_{a \in \A \setminus \R(\1[s])} \EE_{\1[s]}^{a} (\val_{\textnormal{BR}(\X)}(B_1)) < \val_{\textnormal{BR}(\X)}(\1[s]) \,.
    \]
    We show that there is a state in $\L$ with no $\eps$-optimal policy for $\eps$ small enough, which is a contradiction.

    Consider the state $s \in \L$ with the highest belief-reachability value. 
    Note that states with value zero are terminal, so $s$ has strictly positive value.
    Denote the minimum gap obtained by using unreliable actions by 
    \[
        \zeta \defas \min_{\substack{s \in \L,
        \\ a \notin \R(\1[s])}} \val_{\textnormal{BR}(\X)}(\1[s]) - \EE_{\1[s]}^{a} (\val_{\textnormal{BR}(\X)}(B_1)) > 0 \,.
    \]
    Take $\eps \in (0, \min \{\zeta, \val_{\textnormal{BR}(\X)}(\1[s]) \})$ and consider a policy $\policy_\eps$ that is $\eps$-optimal policy for the initial belief $\1[s]$. 
    Note that, because reliable actions do not connect $\L$ to $\T$ and $\val_{\textnormal{BR}(\X)}(\1[s]) > 0$, the policy $\policy_\eps$ must necessarily use unreliable actions.
    We show that $\policy_\eps$ is not $\eps$-optimal, leading to a contradiction.
    
    Denote the first time the action played has positive probability of leaving the set $\L$ by
    \[
        \tau \defas \inf \left\{ t \ge 0 \;:\; \supp \bigl( \transition(S_t, A_t) \bigr) \notin \L \times \Z \right\} \,.
    \]
    Then, we bound the belief-reachability guaranteed by $\policy_\eps$ as follows.
    \begin{align*}
        \PP^{\policy_\eps}_{\1[s]}( \exists t \in \NN \quad B_t \in \D_\X )
            &= \EE_{\1[s]}^{\sigma_\eps} \left( \val_{\textnormal{BR}(\X)}(B_{\tau + 1}) \1[\tau < \infty] \right) 
                & ( \X \cap \L = \emptyset ) \\
            &\le \EE_{\1[s]}^{\sigma_\eps} \left( \val_{\textnormal{BR}(\X)} \left( \1[S_{\tau}] \right) \1[\tau < \infty] \right) - \zeta 
                & (\textnormal{Definition of $\zeta$} ) \\
            &\le \val_{\textnormal{BR}(\X)} \left( \1[s] \right) - \zeta 
                & (\textnormal{Choice of $s$} ) \\
            &< \val_{\textnormal{BR}(\X)}(\1[s]) - \eps \,,
                & (\textnormal{$\eps < \zeta$} )
    \end{align*}
    which contradicts the fact that $\sigma_\eps$ is $\eps$‑optimal. 
    We conclude that $\S_{|\S|} = \S$.

    We now prove that the policy $\policy$ is stopping with parameters $n = |\S| + 2$ and $q = \transition_{\min}^2 \left( \frac{\transition_{\min}}{|\A|} \right)^{|\S|}$.
    Note that, because $\POMDP$ is revealing, for all $t \in [|\S|]$, 
    \[
        s \in \S_t 
        \quad \Rightarrow \quad
        \exists \next{s} \in \S_{t - 1}, a \in \R(\1[s]) \quad \transition(s,a)(\next{s},\next{s}) > 0 \,.
    \]
    By recursion and because $\policy$ plays all reliable actions uniformly at random, we conclude that, for all $s \in \S$,
    \[
        \PP^\policy_{\1[s]} \left( \exists t \le |\S| \quad S_t \in \T \right) 
            \ge \left( \frac{\transition_{\min}}{|\A|} \right)^{|\S|} \,.
    \]
    and thus, 
    \[
        \PP^\policy_{\1[s]} \left( \exists t \le |\S| + 1 \quad B_t \in \D_\T \right) 
            \ge \transition_{\min}\left( \frac{\transition_{\min}}{|\A|} \right)^{|\S|} \,.
    \]
    To conclude, for all initial beliefs $b \in \Delta(\S)$, after every action, the probability that the next belief is a Dirac at least $\transition_{\min}$.
\end{proof}

\stoppingapproxreach*
\begin{proof}
    Consider a POMDP with belief-reachability objectives and an $(n, q)$-stopping optimal policy $\policy$.
    Denote the hitting time of the terminal set $\T\subseteq\S$ by 
    \[
        \tau \defas \inf \{ t \geq 0 : B_t \in \D_\T \}.
    \]
    Because $\policy$ is $(n,q)$-stopping, we have that $\PP^\policy_{b_0}( \tau \le n ) \ge q$.
    Similarly, considering blocks of size $n$, for every integer $k \geq 1$, we have that $\PP^\policy_{b_0} ( \tau > k n ) \le (1 - q)^k$.
    Thus, $\tau$ is finite $\PP^\policy_{b_0}$-a.s.
    
    For every $b_0\in \Delta(\S)$, we have that
    \begin{align*}
        \val_{\textnormal{BR}(\X)}(b_0)
            &= \PP_{b_0}^\policy(\exists t \ge 0 \quad B_t \in \D_\X) 
                &\left(\textnormal{Definition of $\val_{\textnormal{BR}(\X)}$}\right) \\
            &= \PP_{b_0}^\policy(B_\tau\in \D_\X) 
                &(\tau < \infty) \\
            &= \PP_{b_0}^\policy(\tau \le T, B_\tau \in \D_\X ) 
                + \PP_{b_0}^\policy( \tau > T, B_\tau \in \D_\X ) 
                &(\textnormal{Partitioning}) \\
            &\le \PP_{b_0}^\policy( \tau \le T, B_\tau \in \D_\X ) 
                + \PP_{b_0}^\policy( \tau > T) 
                &(\textnormal{Inclusion}) \\
            &\le \PP_{b_0}^\policy( \exists t \le T \quad B_t \in \D_\X ) 
                + \PP_{b_0}^\policy( \tau > T) \,.
                &(\textnormal{Definition of $\tau$}) 
    \end{align*}
    We deduce that $\PP^{\policy}_{b_0} ( \tau > T) \le \eps$ by taking $T = n \left\lceil \dfrac{\log(\eps)}{\log(1-q)} \right\rceil$, which concludes the proof.
\end{proof}

\tstepapprox*

\begin{proof}
    Consider a revealing POMDP with a set of target states $\X \subseteq \S$ and a finite horizon $T\in \NN$.
    Denote the corresponding set of terminal states by $\T\subseteq \S$.
    Consider $b \in \Delta(\S)$ and denote the L1-norm by $\|b\|_1 \defas \sum_{s \in \S} |b(s)|$.

    \noindent \textbf{Simplex Discretization:} We first discretize the set of beliefs $\Delta(\S)$. 
    Let $k \ge 1$ and $\eps > 0$. 
    Define the $k$-uniform grid of beliefs
    \begin{align*}
        G_k \defas \Bigl\{ b \in \Delta(\S): &b(s) = \dfrac{n(s)}{k}, n(s) \in \{0, \ldots, k\},\sum_{s = 1}^{|\S|} n(s) = k \Bigr \} \,.
    \end{align*}
    Denote $\Pi_k\colon \Delta(\S)\to G_k$ the projection function defined by 
    \[
        \Pi_k(b)\defas\argmin_{\substack{b^\prime\in G_k\\[2pt]\supp(b)=\supp(b^\prime)}}\|b-b^\prime\|_1.
    \]
    By construction, it holds that, for every $b\in \Delta(\S)$, there exists $\Pi_k(b) \in G_k$ with 
    \[
        \|b - \Pi_k(b) \|_1 \le \dfrac{|\S|}{k},
    \]
    and taking $k = \left \lceil \dfrac{(T + 1)|\S|}{\eps} \right \rceil$, we have $\|b-\Pi_k(b)\|_1\le \dfrac{\eps}{T+1}$.

    \noindent \textbf{Bellman Equations:} We now define the Bellman equations for the $T$-step belief-reachability value and then for the $T$-step belief-reachability value on the grid previously defined. 
    For every $t \le T$ and $b\in \Delta(\S)$, define the bellman equation of the $t$-step belief reachability value in $\POMDP$ by
    \begin{itemize}
        \item Base case: we have that 
        \[
        \val_{\textnormal{BR}(\X)}^0(b)\defas\left\{
        \begin{array}{ll}
            1 & \textnormal{ if } b=\1[s] \textnormal{ for some } s\in \X,\\
            0 & \textnormal{ otherwise.}
        \end{array}
        \right.
        \]
        \item Induction case: we have that 
        \begin{align*}
            \val_{\textnormal{BR}(\X)}^t(b)& \defas
            \1[b\in \D_\X]+(1-\1[b\in \D_\X])\max_{a\in \A}\sum_{z\in \Z}\sum_{s,s^\prime\in \S} b(s)\transition(s,a)(s^\prime,z)\val_{\textnormal{BR}(\X)}^{t-1}(\tau(b,a,z)),
        \end{align*}
        where $\tau(b,a,z)(s')\defas\frac{\sum_{s\in \S}b(s)\transition(s,a)(s^\prime,z)}{\sum_{s,s^\prime\in \S}b(s)\transition(s,a)(s^\prime,z)}$.
    \end{itemize}
    \noindent We now define the value iteration that considers only beliefs on the grid $G_k$. More formally, for every $b\in \Delta(\S)$, there exists $\Pi_k(b)\in G_k$ such that, 
    \begin{itemize}
        \item Base case: we have that 
        \[
        \widetilde{\val}_{\textnormal{BR}(\X)}^0(\Pi_k(b))\defas\left\{
        \begin{array}{ll}
            1 & \textnormal{ if } \Pi_k(b)=\1[s] \textnormal{ for some } s \in \X,\\
            0 & \textnormal{ otherwise.}
        \end{array}
        \right.
        \]
        \item Induction case: we have that 
        \[
            \widetilde{\val}_{\textnormal{BR}(\X)}^t(\Pi_k(b))\defas \1[\Pi_k(b)\in \D_\X]+(1-\1[\Pi_k(b)\in \D_\X])\max_{a\in \A}\EE_{\Pi_k(b)}^{a}\left(\widetilde{\val}_{\textnormal{BR}(\X)}^{t-1}(\Pi_k(B_1))\right).
        \]
    \end{itemize}
    
    \noindent \textbf{Lipschitz Property:} We show that the $T$-step belief-reachability value is a Lipschitz function over beliefs with same support.
    For all beliefs $b,b^\prime\in \Delta(\S)$ with $\supp(b)=\supp(b^\prime)$, we have
    \begin{align*}
        &\left| \val_{\textnormal{BR}(\X)}^T(b) - \val_{\textnormal{BR}(\X)}^T(b^\prime) \right|\\
        &\qquad = \left| \max_{a \in \A} \EE_b^a \left(\val_{\textnormal{BR}(\X)}^{T-1}\left(B_1\right)\right)- \max_{a \in \A} \EE_{b^\prime}^a \left(\val_{\textnormal{BR}(\X)}^{T-1}\left(B_1\right)\right) \right| & (\textnormal{Bellman optimality})\\
        &\qquad \le \left|\max_{a \in \A} \left[\EE_b^a \left(\val_{\textnormal{BR}(\X)}^{T-1}\left(B_1\right)\right) - \EE_{b^\prime}^a \left(\val_{\textnormal{BR}(\X)}^{T-1}\left(B_1\right)\right)\right] \right|&(\text{Subadditivity of $\max$})\\
        &\qquad \le \max_{a \in \A} \left| \EE_b^a \left(\val_{\textnormal{BR}(\X)}^{T-1}\left(B_1\right)\right) - \EE_{b^\prime}^a \left(\val_{\textnormal{BR}(\X)}^{T-1}\left(B_1\right)\right) \right|&(\text{Monotonicity of $\max$})\\
        &\qquad = \max_{a \in \A} \left| \sum_{s \in \S} \bigl(b(s) - b^\prime(s)\bigr) \EE_{1[s]}^a \left(\val_{\textnormal{BR}(\X)}^{T-1}\left(B_1\right)\right)  \right| & \left(\textnormal{Affinity of } \EE_b^a \left(\val_{\textnormal{BR}(\X)}^{T-1}\left(B_1\right)\right) \right)\\
        &\qquad \le \sum_{s \in \S} \left|b(s) - b^\prime(s) \right| = \left\|b - b^\prime \right\|_1\,, & \left(0 \le \EE_{1[s]}^a \left(\val_{\textnormal{BR}(\X)}^{T-1}\left(B_1\right)\right) \le 1\right)
    \end{align*}

    \noindent\textbf{Error Bound:}
    Take $k=\left\lceil\frac{(T+1)|\S|}{\eps}\right\rceil$. 
    Define for every $t\in \{0,\ldots,T\}$ and belief $b\in \Delta(\S)$,
    \[
        E_t(b)\defas\left|\val_{\textnormal{BR}(\X)}^{t}(b)-\widetilde{\val}_{\textnormal{BR}(\X)}^{t}(\Pi_k(b))\right|
    \]
    We prove using an induction argument that, for every $t \le T$ and $b\in \Delta(\S)$, 
    \[
        E_t(b)\le \dfrac{(t+1)\eps}{T+1}.
    \]
    Observe that when $t=0$, we have that $E_0(b)\le\dfrac{\eps}{T+1}$. 
    Now, assume that $E_{t-1}\le \frac{t\eps}{T+1}$.
    Therefore, we have that, for every $b\in \Delta(\S)$ and taking $k=\left\lceil\dfrac{(T+1)|\S|}{\eps}\right\rceil$,
    \begin{align*}
        E_t(b) &=\left|\val_{\textnormal{BR}(\X)}^t(b)-\widetilde{\val}_{\textnormal{BR}(\X)}^t(\Pi_k(b))\right|\\
        & \le \max_{a\in \A}\left|\EE_{b}^a\left(\val_{\textnormal{BR}(\X)}^{t-1}(B_1)\right)-\EE_{\Pi_k(b)}^a\left(\widetilde{\val}_{\textnormal{BR}(\X)}^{t-1}\left(\Pi_k(B_1)\right)\right)\right|\\
        & \le \max_{a\in \A}\left|\EE_{b}^a\left(\val_{\textnormal{BR}(\X)}^{t-1}(B_1)\right)-\EE_{\Pi_k(b)}^a\left(\val_{\textnormal{BR}(\X)}^{t-1}\left(B_1\right)\right)\right|\\
            &\qquad+\max_{a\in \A}\left|\EE_{\Pi_k(b)}^a\left(\val_{\textnormal{BR}(\X)}^{t-1}\left(B_1\right)\right)-\EE_{\Pi_k(b)}^a\left(\widetilde{\val}_{\textnormal{BR}(\X)}^{t-1}\left(\Pi_k(B_1)\right)\right)\right|\\
        & \le \left\|b-\Pi_k(b)\right\|_1 &\textnormal{(Lipschitz property)}\\
            &\qquad + \max_{a\in \A}\left|\EE_{\Pi_k(b)}^a\left(\val_{\textnormal{BR}(\X)}^{t-1}\left(B_1\right)\right)-\EE_{\Pi_k(b)}^a\left(\widetilde{\val}_{\textnormal{BR}(\X)}^{t-1}\left(\Pi_k(B_1)\right)\right)\right|\\
        & \le \left\|b-\Pi_k(b)\right\|_1\\
            &\qquad + \max_{a\in \A}\EE_{\Pi_k(b)}^a\left(\left|\val_{\textnormal{BR}(\X)}^{t-1}\left(B_1\right)-\widetilde{\val}_{\textnormal{BR}(\X)}^{t-1}\left(\Pi_k(B_1)\right)\right|\right)\\
        & \le  \dfrac{\eps}{T+1} + \max_{a\in \A}\EE_{\Pi_k(b)}^a\left(\left|\val_{\textnormal{BR}(\X)}^{t-1}\left(B_1\right)-\widetilde{\val}_{\textnormal{BR}(\X)}^{t-1}\left(\Pi_k(B_1)\right)\right|\right) &\text{(Definition of $k$)}\\
        & \le \dfrac{\eps}{T+1}+\max_{a\in \A}\EE_{b}^a(E_{t-1}(B_1))&\text{(Definition of $E_{t-1}$)}\\
        & \le \dfrac{\eps}{T+1}+\dfrac{t\eps}{T+1}&\textnormal{(Induction hypothesis)}\\
        & = \dfrac{(t+1)\eps}{T+1},
    \end{align*}
    and the induction argument holds. Therefore, we obtain that $E_T(b)\le \frac{(T+1)\eps}{T+1}=\eps$.\\

    \noindent\textbf{Complexity:}
    For every action, the bellman update considers $O(|\S|^2|\Z|)$. Taking the maximum over actions, we get a complexity bound of $O(|\S|^2|\A||\Z|)$.
    We considered the update on a uniform grid $G_k\subseteq \Delta(\S)$ with, for all $b\in\Delta(\S)$, $\|b-\Pi_k(b)\|_1\le \eps/(T+1)$ by taking $k=\left\lceil\dfrac{(T+1)|\S|}{\eps}\right\rceil$.
    The size of grid in the $(|\S|-1)$-simplex $\Delta(\S)$ is thus 
    \[
    \left|G_k\right| = \binom{k+|\S|-1}{|\S|-1}=O \left( \left( \tfrac{T|\S|}{\eps} \right)^{|\S|-1} \right).
    \] 
    Therefore, we deduce that
    \begin{align*}
        \textnormal{Total time complexity}&=O\left((T+1)\cdot|G_k|\cdot |\S|^2||\A||\Z|\right)\\
        &=O\left((T+1)\cdot\left(\dfrac{T|\S|}{\eps}\right)^{|\S|-1}\cdot|\S|^2||\A||\Z|\right)\\
        &=O\left(\dfrac{T^{|\S|}|\S|^{|\S|+1}|\A||\Z|}{\eps^{|\S|-1}}\right),
    \end{align*}
    which completes the proof. 
\end{proof}

\approxbeliefreach*

\begin{proof}
    Let $\POMDP$ be a revealing POMDP with initial belief $b\in \Delta(\S)$ and set of target states $\X\subseteq \S$.
    By \Cref{Result: Existence of optimal policy}, there exists a $(n,q)$-stopping optimal policy $\policy$ with $n=|\S|+2$ and $q=\transition_{\min}^2\left(\tfrac{\transition_{\min}}{|\S|}\right)^{|\S|}$. 
    By \Cref{Result: Stopping optimal policy approximate the value}, for every $\eps>0$, there exists a time horizon $T=\left\lceil\tfrac{n\log(\eps)}{\log(1-q)}\right\rceil$ such that 
    \[
        \PP_b^\policy(\exists t \le T \quad B_t\in \D_\X)\ge \val_{\textnormal{BR}(\X)}(b)-\eps.
    \]
    Moreover, by \Cref{Result: Approximation of the T step reachability value}, we can compute an approximation of $\val_{\textnormal{BR}(\X)}^{\,T}(b)$ denoted $\widetilde{\val}_{\textnormal{BR}(\X)}^{\,T}$ with total time complexity 
    \[
        O\left(\dfrac{T^{|\S|}|\S|^{|\S|+1}|\A||\Z|}{\eps^{|\S|-1}}\right).
    \]
    We consider the grid $G_k$ as previously defined in the proof of \Cref{Result: Approximation of the T step reachability value} with $k=\left\lceil\tfrac{(T+1)|\S|}{\eps}\right\rceil$.
    For every $b\in \Delta(\S)$, there exists $\Pi_k(b)\in G_k$ such that by taking $T=\left\lceil\tfrac{n\log(\eps)}{\log(1-q)}\right\rceil$, 
    \begin{align*}
        &\left|\val(b)_{\text{BR}(\X)}-\widetilde{\val}_{\text{BR}(\X)}^{\,T}\left(\Pi_k(b)\right)\right|\\
        &\qquad \le \left|\val(b)_{\text{BR}(\X)}-\PP_b^\policy(\exists t \le T \quad B_i\in \D_\X)\right| + \left|\PP_b^\policy(\exists t \le T \quad B_t\in \D_\X)-\widetilde{\val}_{\text{BR}(\X)}^{\,T}\left(\Pi_k(b)\right)\right|&\text{(Triangle inequality)}\\
        &\qquad \le \eps + \left|\PP_b^\policy(\exists t \le T \quad B_t\in \D_\X)-\widetilde{\val}_{\text{BR}(\X)}^{\,T}\left(\Pi_k(b)\right)\right| & \text{(\Cref{Result: Stopping optimal policy approximate the value})}\\
        &\qquad \le 2\eps. &\text{(\Cref{Result: Approximation of the T step reachability value})}
    \end{align*}
    Let $L\in \NN$ be the bit-length bound on transition probabilities so that $\transition_{\min}\geq 2^{-L}$. Since $q\in (0,1]$ and using that $\log(1-q)\le -q$, we have that
    \begin{align*}
        T
        &\le n\left(\dfrac{\log(\eps)}{\log(1-q)}+1\right)\\
        &\le n\left(\dfrac{1}{q}\log\left(\dfrac{1}{\eps}\right)+1\right)\\
        &\le n\left(\left(\dfrac{|\A|}{\transition_{\min}}\right)^{|\S|}\dfrac{1}{\transition_{\min}^2}
        \log\left(\dfrac{1}{\eps}\right)+1\right)\\
        &\le n\left(\left(|\A| 2^L\right)^{|\S|} 2^{2L}\log\left(\dfrac{1}{\eps}\right)+1\right),
    \end{align*}
    and thus, we deduce the following upper bound for the horizon
    \[
    T=O\left((|\S|+2)2^{L(|\S|+2)}|\A|^{|\S|}\log\left(\dfrac{1}{\eps}\right)\right)
    \]
    and thus, $T^{|\S|}=O\left((|\S|+2)^{|\S|}2^{L|\S|(|\S|+2)}|\A|^{{|\S|}^2}\log(1/\eps)^{|\S|}\right)$.
    Inserting this bound into the above complexity yields 
    \begin{align*}
        \textnormal{Total time complexity}=O\left(\dfrac{(|\S|+2)^{|\S|}|\S|^{|\S|+1}|\Z|2^{L|\S|(|\S|+2)}|\A|^{{|\S|}^2+1}\log(1/\eps)^{|\S|}}{\eps^{|\S|-1}}\right).
    \end{align*}
    Therefore, the total time complexity is upper bounded by $2^{\textnormal{poly}(L,|\S|,\log(|\A|),\log(|\Z|),\log(1/\eps))}$, which concludes the proof.
\end{proof}


\section{Proofs of {\Cref{Section: Parity}}}

\almostsuresafetygeneral*

\begin{proof}
    Consider a POMDP $\POMDP = (\S, \A, \Z, \transition, b_0)$, its underlying MDP $\MDP$ with states $\S \times \Z$, an end component $\U = (\Q, \E)$ of $M$, and a policy $\policy$ on $\POMDP$ such that $\PP_{b_0}^{\policy} ( \I(\play) = \U ) > 0$.
    By contradiction, assume that there exists a state $q \in \Q$, an action $a \in \E(q)$, and a signal $z \in \Z$ such that, for all policies $\another{\policy}$ on $\POMDP$, we have
    \begin{equation}
    \label{Assumption: Parity less than 1}
        \PP_{b}^{\another{\policy}} ( \I(\play) \subseteq \U) < 1 \,,
    \end{equation}
    where $b$ is the belief after starting with belief $\1[q]$, playing action $a$, and receiving signal $z$.
    We show that $\PP_{b_0}^{\policy} ( \I(\play) = \U) = 0$, which is a contradiction.

    By definition of end components, because $\PP_{b_0}^{\policy} ( \I(\play) = \U ) > 0$, for every action $a \in \E(q)$, the policy $\policy$ plays action $a$ infinitely often when the belief of the controller has positive probability on $q$.
    By \Cref{Assumption: Parity less than 1}, for all policies $\another{\policy}$ on $\POMDP$, the set $\S \setminus \Q$ is reached with positive probability, i.e., 
    \begin{equation}
    \label{Statement: Positive reachability}
        \PP_{b}^{\another{\policy}} ( \exists t \ge 0 \quad S_t \in \S \setminus \Q ) > 0 \,.
    \end{equation}
    We claim that the set $\S \setminus \Q$ is in fact reached within $2^{|\Q|}$ steps with positive probability.
    Formally, for all policies $\another{\policy}$,
    \begin{equation}
    \label{Eq: Leaving probability within bounded time horizion}
        \PP_{b}^{\another{\policy}} ( \exists t \le 2^{|\Q|} \quad S_t \in \S \setminus \Q) > 0 \,.
    \end{equation}
    
    Indeed, we show that, starting from $b$, all policies reach a belief at which every action leads outside of $\Q$ with positive probability. 
    Consider a belief $\another{b} \subseteq \Delta(\Q)$. 
    Then, an action $a \in \A$ is safe for $\another{b}$ if $\PP_{\another{b}}^{a} (\supp(B_1) \subseteq \Q) = 1$.
    Note that, if an action is safe for a belief, then it is safe for all beliefs with the same support.
    Also, starting from $b$, if the controller can force the dynamic to visit only beliefs whose support has some safe action, then taking these actions one constructs a policy $\another{\policy}$ such that $\PP_{b}^{\another{\policy}} ( \forall t \ge 0 \quad S_t \in \Q) = 1$.
    By \Cref{Statement: Positive reachability}, this is not the case, i.e., for all policies the dynamic must visit a belief whose support has no safe action.
    Moreover, for each policy, such a belief is reached in the first $(2^{|\Q|} - 1)$ steps because there are only so many supports, which proves our claim.

    \Cref{Eq: Leaving probability within bounded time horizion} implies that the probability of reaching the set $\S \setminus \Q$ is lower bounded, i.e., for all policies $\another{\policy}$,
    \[
        \PP_{b}^{\another{\policy}} \left( \exists t \le 2^{|\Q|} \quad S_t \in \S \setminus \Q \right ) 
        \;\ge\; \left(\transition_{\min} \right)^{2^{|\Q|}} \,.
    \]
    With this, we prove that $\PP_{b_0}^{\policy} ( \I(\play) = \U) = 0$, which is a contradiction.
    
    Indeed, 
    because $\1[q]$ is visited infinitely often, action $a$ played infinitely often, and therefore the belief $b$ is visited infinitely often.  
    Conditioned on the event $\{\inf(\play)=\U\}$, visits to the state $q$ followed by playing the action $a$ occur infinitely
    often under policy $\policy$. The above inequality implies that some states $s \in \S \setminus \Q$ are reached with lower-bounded positive probability after such visits. 
    Therefore, we get
    \[
        \PP_{b_0}^{\policy} ( \I(\play) = \U) = 0\,,
    \]
    which contradicts with the assumption, and therefore, yields the result.
\end{proof}

\almostsuresafetyrevealing*

\begin{proof}
    By \Cref{Result: Existence of almost-sure saftey policy for an end component in general PODMPs}, for all states $q \in \Q$, actions $a \in \E(q)$, and signals $z \in \Z$, there exists a policy $\policy_{q,a,z}$ such that 
    \[
        \PP_{b}^{\policy_{q,a,z}} ( \I(\play) \subseteq \U) \;=\; 1 \,,
    \]
    where $b$ is the belief after starting with belief $\1[q]$, playing action $a$, and receiving signal $z$.
    We now construct a policy $\policy_q$ such that
    \[
        \PP_{\1[q]}^{\policy_q} ( \I(\play) = \U) \;=\; 1 \,.
    \]
    The policy proceeds in \emph{phases}, one for every ordered pair of state $\another{q} \in \Q$ and $\another{a} \in \E(\another{q})$. The goal of phase $(\another{q}, \another{a})$ is to reach the Dirac belief $\1[\another{q}]$ and play action $\another{a}$. By cycling through all pairs, each possible state–action pair is visited infinitely often. For each phase $(\another{q}, \another{a})$, the policy $\policy_q$ is as follows. While the current belief is \emph{not} Dirac,
    $\policy_q$ follows the almost‑sure safety policy defined in \Cref{Result: Existence of almost-sure saftey policy for an end component in general PODMPs}. Since the POMDP is revealing, a Dirac belief $\1[\next{q}]$ is reached almost-surely. 
    Since $\U$ is an end component in the MDP, there exists a finite path $(q_0, a_0, q_1, a_1 \ldots, q_k)$ inside $\U$ such that $q_0 = \next{q}$ and $q_k = \another{q}$, and $q_{i+1} \in \supp(\transition(q_i, a_i))$. 
    The policy now plays $(a_0, a_1, \dots ,a_{k-1})$ in order, stopping early if the belief becomes non‑Dirac. In this case, it returns to the safety policy to stay inside $\U$. Once the belief is $\1[\another{q}]$, $\policy_q$ plays the action $\another{a}$ and then proceeds to the next phase. Therefore, all states and actions inside end component $\U$ are visited infinitely often, which completes the proof.
\end{proof}

\paritytobeliefreach*

\begin{proof}
    Consider a revealing POMDP $P$ with the initial belief $b_0 \in \Delta(\S)$. 
    Denote the set of states for which, if the initial belief were a Dirac on that state, then the parity condition can be satisfied almost-surely by $\X \defas \left\{ s \in \S : \exists \policy \quad \PP_{\1[s]}^{\policy}( \Parity ) = 1 \right\}$.   
    We show that, for all policies $\policy\in \Policies$, 
    \[
        \PP_{b_0}^{\policy} \bigl ( \play \in  \Parity \bigr ) 
            = \PP_{b_0}^{\policy} \bigl ( \exists t \ge 0 \quad B_t \in \D_{\X} \bigr ) \,,
    \]
    which implies the equality by maximizing over the policy.
    
    Fix a policy $\policy$.
    Denote $\W$ the set of all end components of the underlying MDP on $\S \times \Z$ that can be reached by $\policy$.
    Formally,
    \[
        \W \defas \{ \U : \PP_{b_0}^{\policy}( \I(\play) = U ) > 0 \} \,. 
    \]
    Fix an end component $\U = (\Q, \E) \in \W$. 
    Note that, if a play $\play$ eventually remains in $U$ forever, i.e., $\I(\play) = \U$, then, by definition of end component, all the states $q \in \Q$ are visited infinitely often.
    Therefore, if $\I(\play) = \U$, then
    \[
        \play \in \Parity \iff \min \{ \priority(q) : q \in \Q \} \textnormal{ is even.}
    \]
    Consequently, the parity condition being satisfied depends on the end component reached. 
    Formally, for all end components $\U \in \W$,
    \[
        \PP_{b_0}^{\policy} \bigl( \play \in \Parity \,\mid\, \I(\play) = \U \bigr) \in \{0, 1\}\,.
    \]
    Denote by $\W_{\textnormal{P}} \subseteq \W$ the set of end components where the parity condition is satisfied with probability $1$ conditioned on eventually remaining in the end component.
    Formally, 
    \[
        \W_{\textnormal{P}}
            \defas \left\{ \U \in \W :
                \PP_{b_0}^{\policy} \bigl( \play \in \Parity \,\mid\, \I(\play) = \U \bigr) = 1 
            \right\} \,.
    \]

    Recall that, by \Cref{Result: End components}, some end component must be reached.
    Fix an arbitrary end component $\U = (\Q, \E) \in \W_{\textnormal{P}}$. 
    By \Cref{Result: Existence of almost-sure saftey policy for an end component}, for all states $q \in \Q$, there exists an almost-sure winning policy for the parity objective when the initial belief is $\1[q]$. 
    Therefore, $q$ is an almost-sure winning parity state. 
    In conclusion, we get that
    \begin{align*}
        \PP_{b_0}^{\policy} \bigl ( \play \in  \Parity \bigr ) 
            &= \sum_{\U \in \W_{\textnormal{P}}} \PP_{b_0}^{\policy} \bigl ( \I(\play) = \U \bigr ) 
                & (\textnormal{\Cref{Result: End components}})\\
            &= \PP_{b_0}^{\policy} \bigl ( \exists t \ge 0 \quad B_t \in \D_{\X} \bigr ) \,,
                & (\textnormal{\Cref{Result: Existence of almost-sure saftey policy for an end component}})
    \end{align*}
    which yields the result.
\end{proof}

\approxparity*

\begin{proof}
    By \cite[Theorem 3]{bellyRevelationsDecidableClass2025}, computing almost-sure winning parity states in the revealing POMDP is in \textnormal{EXPTIME}.
    By \Cref{Result: Approximating belief-reachability is EXPTIME}, approximating the belief-reachability value to this set of states is in \textnormal{EXPTIME}.
    By \Cref{Result: Reduction of parity to belief-reachability}, this value coincides with the parity value and therefore we computed an approximation of the parity value in \textnormal{EXPTIME}.
\end{proof}

\parityoptimal*

\begin{proof}
    Let $\S_{\textnormal{P}}$ be the set of almost-sure winning parity states.
    By \Cref{Result: Reduction of parity to belief-reachability}, the parity value coincides with the belief-reachability value to the set $\S_{\textnormal{P}}$. 
    By \Cref{Result: Optimal policy for reachability}, there exists an optimal belief-reachability policy $\policy_0$. 
    For all state $s \in \S_{\textnormal{P}}$, let $\policy_{s}$ be the almost-sure winning parity policy when the initial belief is $\1[s]$. 
    We now construct the optimal policy $\policy$ as follows.
    The optimal policy follows the policy $\policy_0$ until it reaches a Dirac belief $\1[s]$ where $s \in \S_{\textnormal{P}}$. Then, it switches to following the policy $\policy_{s}$. 
    The probability of satisfying the parity condition under the policy $\policy$ is the belief-reachability value to the set $\S_{\textnormal{P}}$, which coincides with the parity value. This completes the proof.
\end{proof}

\limitsure*

\begin{proof}
    Consider revealing POMDPs with parity objectives.
    By \Cref{Result: Revealing POMDPs with parity objectives have 0-optimal policy}, optimal policies exist.
    Therefore, if a revealing POMDP with parity objectives is limit-sure winning, then it is almost-sure winning.
    The reverse direction follows by definition.
    Lastly, by~\cite[Theorem 3]{bellyRevelationsDecidableClass2025}, the existence of an almost-sure policy is in EXPTIME-complete. 
    Therefore, limit-sure analysis is also EXPTIME-complete because the two problems coincide.
\end{proof}

\section{Outline of Algorithms}

This section outlines three algorithms for quantitative analysis: for the $T$-step belief-reachability, for the belief-reachability, and for the parity value, respectively. 

Given a belief $b\in \Delta(\S)$, an action $a\in \A$, and a signal $z\in \Z$, recall that the belief update is defined by
\[
\tau(b,a,z)(s')\defas\frac{\sum_{s\in \S}b(s)\transition(s,a)(s^\prime,z)}{\sum_{s,s^\prime\in \S}b(s)\transition(s,a)(s^\prime,z)}\,.
\]

\paragraph{Outline of \Cref{Algorithm: Approximation of T-step belief-reachability}}
Given a revealing POMDP $\POMDP$, a target set $\X \subseteq \S$, a time horizon $T\in \NN$, and an additive error $\eps>0$, the algorithm returns a value that is guaranteed to be within $\eps$ of the $T$-step belief-reachability value.
Intuitively, \Cref{Algorithm: Approximation of T-step belief-reachability} replaces the continuous set of beliefs with a finite uniform grid to run value iteration only on those grid points.

The algorithm has two stages and follows from the proof of \Cref{Result: Approximation of the T step reachability value}:
\begin{itemize}
    \item[$(1)$] 
        We build a finite uniform grid of beliefs $G_k\subseteq\Delta(\S)$ such that every belief is at most $\eps/(T+1)$ away from some grid point.
    \item[$(2)$] 
        We perform $T$ backward Bellman updates on the grid: each grid belief evaluates every action and keeps the one with the best expected value over the projected successor beliefs.
\end{itemize}
Because the $T$-step belief-reachability value is Lipschitz, the projection errors accumulated over the steps is at most $\varepsilon$.
Therefore, the algorithm returns an epsilon approximation of the $T$-step belief-reachability value.

\paragraph{Outline of \Cref{Algorithm: Approximate belief-reachability value}}
Given a revealing POMDP $\POMDP$, a target set $\X\subseteq \S$, and an additive error $\eps>0$, the algorithm return a value that is guaranteed to be within $\eps$ of the belief-reachability value.
Intuitively, \Cref{Algorithm: Approximate belief-reachability value} approximates the belief-reachability value by computing the $T$-step belief reachability value for a horizon $T$ chosen large enough that the two differ by at most $\eps$.

The algorithm proceeds as follows:
\begin{itemize}
    \item[$(1)$] 
        We first choose $T$ large enough.
        By \Cref{Result: Stopping and reliable is optimal} and \Cref{Result: Existence of optimal policy}, there is an $(n,q)$-stopping optimal policy $\policy$ with parameters $n=|\S|+2$ and $q=\delta_{\min}^2(\delta_{\min}/|\A|)^{|\S|}$ for the belief-reachability objectives. 
        Next, by \Cref{Result: Stopping optimal policy approximate the value}, we set $T=n\lceil\log(\eps/2)/\log(1-q)\rceil$ so that the $T$-step belief-reachability objective under $\policy$ is an $\eps/2$ approximation of the belief-reachability value.
    \item[$(2)$] 
        We call \Cref{Result: Approximation of the T step reachability value} by taking $T=n\lceil\log(\eps/2)/\log(1-q)\rceil$ and $\eps/2$. 
\end{itemize}
By \Cref{Result: Approximating belief-reachability is EXPTIME}, the algorithm returns an epsilon approximation of the belief-reachability value.

\paragraph{Outline of \Cref{Algorithm: Approximation of parity}}
Given a revealing POMDP $\POMDP$, a priority function $\priority: \S \to \{0,\dots,d\}$, and an additive error $\eps>0$, the algorithm returns a value that is guaranteed to be within $\eps$ of the parity value.
Intuitively, \Cref{Algorithm: Approximation of T-step belief-reachability} computes the approximation of the belief-reachability objectives to the set of Dirac on the almost-sure winning parity states.

The algorithm proceeds as follows:
\begin{itemize}
    \item[$(1)$]
        We use \cite[Theorem 3]{bellyRevelationsDecidableClass2025} to compute the set of almost-sure winning states, denoted $\S_P$.
    \item[$(2)$] 
        We call \Cref{Algorithm: Approximate belief-reachability value} by taking the target set as $\S_P$.
\end{itemize}
By \Cref{Result: Reduction of parity to belief-reachability}, the parity value of $\POMDP$ and the belief-reachability value to the set of Dirac on the almost-sure winning parity states objective coincide.
Therefore, because \Cref{Algorithm: Approximate belief-reachability value} returns an epsilon approximation of the belief reachability value, the algorithm returns an epsilon approximation of the parity value.

\begin{algorithm}[t]
    \caption{Approximation of $T$-step belief-reachability value}
    \label{Algorithm: Approximation of T-step belief-reachability}
    \KwIn{Revealing POMDP $P=(\S, \A, \Z, \transition, b_0)$, 
          target set $\X \subseteq \S$, horizon $T \in \NN$, and additive error $\eps > 0$}
    \KwOut{$v$ such that 
            $\left|v -\val^{T}_{\textnormal{BR}(\X)}(b_0)\right|\le \eps$}
    
    $k\leftarrow \left \lceil \dfrac{(T+1)|\S|}{\eps}\right \rceil$ \;
    
    $G_k \leftarrow 
       \left \{
         b\in\Delta(S):
         b(s)=\dfrac{n(s)}{k},\;
         n(s)\in\{0,\dots,k\},\,
         \sum_{s} n(s)=k
       \right \}$ \;

    \ForEach{$b\in G_k$}{
        $\widetilde{\val}_{\textnormal{BR}(\X)}^0(b) \leftarrow \begin{cases}
            1 & \text{if } b \in \D_\X, \\
            0 & \text{otherwise.}
        \end{cases}$\;
    } 
    
    \For{$t\gets 1$ \KwTo $T$}{
        \ForEach{$b\in G_k$}{
            \lIf{$b\in \D_\X$}{$\widetilde{\val}_{\textnormal{BR}(\X)}^t(b) \leftarrow 1$}
            \lElse{
                \ForEach{$a\in \A$}{
                    $Q_a\leftarrow 0$ \;
                    \ForEach{$z\in \Z$}{
                        \ForEach{$s,s'\in \S$}{
                            $p\leftarrow b(s)\,\delta(s,a)(s',z)$\;
                            \If{$p>0$}{
                                $b'\leftarrow \Pi_k(\tau(b,a,z))$\;  
                                $Q_a\leftarrow Q_a + p \cdot \widetilde{\val}_{\textnormal{BR}(\X)}^{t-1}(b')$\;
                            }
                        }
                    }
                }
                $ \widetilde{\val}_{\textnormal{BR}(\X)}^t(b) \leftarrow \max_{a\in A} Q_a$ 
            }
        }
    }
    \Return $\widetilde{\val}_{\textnormal{BR}(\X)}^T(\Pi_k(b))$\;
\end{algorithm}

\begin{algorithm}[t]
    \caption{Approximation of belief-reachability value}
    \label{Algorithm: Approximate belief-reachability value}
    \KwIn{Revealing POMDP $P=(\S,\A,\Z,\transition,b_0)$, target set $\X\subseteq\S$, and additive error $\eps>0$}
    \KwOut{$v$ such that $\left|v-\val_{\textnormal{BR}(\X)}(b_0)\right|\le\eps$}
    
    $\delta_{\min}\;\leftarrow\;\min\bigl\{\transition(s,a)(s',z): \forall s,s' \in \S, a \in \A, z \in \Z \quad \transition(s,a)(s',z)>0\bigr\}$\;
    
    $n \;\leftarrow\; |\S|+2$\;

    $T \;\leftarrow\;
          n\,\left\lceil\; \dfrac{\log(\eps/2)}{\log(1-q)} \right\rceil$\; 
    
    $v \;\leftarrow\;
       \textnormal{\Cref{Algorithm: Approximation of T-step belief-reachability}}\bigl(P,\X,b_0,T,\eps/2\bigr)$\;

    \Return $v$\;
\end{algorithm}

\begin{algorithm}[t]
    \caption{Approximation of parity value}
    \label{Algorithm: Approximation of parity}
    \KwIn{Revealing POMDP $P=(\S,\A,\Z,\transition,b_0)$, 
          priority function $\priority: \S \to \{0,\dots,d\}$, 
          and additive error $\eps>0$}
    \KwOut{$v$ such that $\left|v-\val_{\textnormal{P}}(b_0)\right| \le \eps$}
    
    $\X
       \;\leftarrow\;
       \textsc{AlmostSureParityStates}(P,\priority)$\;
    
    $v
       \;\leftarrow\;
       \textnormal{\Cref{Algorithm: Approximate belief-reachability value}}
          \bigl(P,\X,b_0,\eps\bigr)$\;
    
    \Return $v$\;
\end{algorithm}

\end{document}